\documentclass[11pt]{article}
\usepackage[margin=1in]{geometry}                		%
\usepackage{graphicx} %
\usepackage{xcolor}
\definecolor{ForestGreen}{rgb}{0.1333,0.5451,0.1333}
\definecolor{DarkRed}{rgb}{0.65,0,0}
\definecolor{Red}{rgb}{1,0,0}
\usepackage[linktocpage=true,
pagebackref=true,colorlinks,
linkcolor=DarkRed,citecolor=ForestGreen,
bookmarks,bookmarksopen,bookmarksnumbered, hypertexnames=false]
{hyperref}
\usepackage{comment}
\usepackage{enumerate}
\usepackage{multirow}
\usepackage{array}

\usepackage[margin=1in]{geometry}                		%
\usepackage{amssymb}
\usepackage{amsthm}
\usepackage{amsmath}
\usepackage{thm-restate}

\usepackage{algorithm}
\usepackage{algpseudocode}
\usepackage{float}
\usepackage{tablefootnote}
\usepackage{longtable}

\usepackage{xspace}
\usepackage[colorinlistoftodos,textsize=tiny,textwidth=2cm,color=red!25!white,obeyFinal]{todonotes}
\usepackage{cleveref}

\newtheorem{theorem}{Theorem}[section]

\newtheorem{lemma}[theorem]{Lemma}
\newtheorem{observation}[theorem]{Observation}
\newtheorem{proposition}[theorem]{Proposition}
\newtheorem{definition}[theorem]{Definition}
\newtheorem{corollary}[theorem]{Corollary}
\newtheorem{claim}[theorem]{Claim}

\newcommand{\poly}{\mathrm{poly}\xspace}

\newcommand{\CC}{\mathcal{C}}

\newcommand{\OO}{O}

\global\long\def\src{{\rm src}}
\global\long\def\sink{{\rm sink}}

\global\long\def\Decomp{\textsc{Decomp}}

\global\long\def\iso{{\rm iso}}
\global\long\def\wit{{\rm wit}}
\global\long\def\src{{\rm src}}
\global\long\def\sink{{\rm sink}}

\global\long\def\local{{\rm local}}

\newcommand{\wtilde}{\widetilde}

\title{Unbreakable Decomposition in Close-to-Linear Time}

\author{
Aditya Anand\thanks{University of Michigan, Ann Arbor} 
\and Euiwoong Lee\thanks{University of Michigan, Ann Arbor. Supported in part by NSF grant CCF-2236669 and Google.}
\and Jason Li\thanks{Carnegie Mellon University}
\and Yaowei Long\thanks{University of Michigan, Ann Arbor}
\and Thatchaphol Saranurak\thanks{University of Michigan, Ann Arbor. Supported by NSF Grant CCF-2238138.}
}
\date{}

\begin{document}

\maketitle
\pagenumbering{gobble}
\begin{abstract}
    Unbreakable decomposition, introduced by \cite{cygan2014minimum,cygan2020randomized}, 
has proven to be one of the most powerful tools for parameterized graph cut problems in recent years.
Unfortunately, all known constructions require at least $\Omega_k\left(mn^2\right)$ time, given an undirected graph with $n$ vertices, $m$ edges, and cut-size parameter $k$.
In this work, we show the first close-to-linear time parameterized algorithm that computes an {unbreakable decomposition}. More precisely, for any $0<\epsilon\leq 1$, our algorithm runs in time $2^{\OO(\frac{k}{\epsilon} \log \frac{k}{\epsilon})}m^{1 + \epsilon}$ and computes a $(\OO(k/\epsilon), k)$ unbreakable tree decomposition of $G$, where each bag has adhesion at most $\OO(k/\epsilon)$.

This immediately opens up possibilities for obtaining close-to-linear time algorithms for numerous problems whose only known solution is based on unbreakable decomposition.

\end{abstract}

\newpage
\tableofcontents
\newpage

\pagenumbering{arabic}

\section{Introduction}

For the past two decades, the study of graph cut problems has been a highly active subarea of fixed-parameter tractability (FPT) algorithms that has led to many powerful algorithmic techniques, including important separators and shadow removal \cite{chen2008fixed,marx2011fixed,chitnis2013fixed,lokshtanov2013clustering}, matroid-based kernelization \cite{kratsch2012representative,kratsch2014compression,chalermsook2021vertex,he2021near,wahlstrom2022quasipolynomial}, treewidth reduction \cite{marx2013finding}, branching from half-integral solutions \cite{cygan2013multiway,guillemot2011fpt,lokshtanov2014faster,iwata2016half,iwata2016linear,iwata20180}, and flow augmentation \cite{kim2021flow,kim2024flow,kim2024weighted}.

Recently, \emph{unbreakable decomposition} has emerged as one of the most powerful techniques used by numerous FPT algorithms \cite{cygan2014minimum,cygan2020randomized,lokshtanov2022parameterized,pilipczuk2021algorithms,schirrmacher2024model,saurabh202364,inamdar2023parameterized,lokshtanov2019balanced,agrawal2022fixed}. This decomposition also generalizes the highly influential tree decomposition of bounded width to general graphs, in the sense that unbreakable decomposition exists on arbitrary graphs, and the two concepts coincide on bounded-treewidth graphs. See \Cref{sec:connection}.

We briefly define unbreakable decomposition here. Given a graph $G$ and a vertex set $X$, a \emph{$(q,k)$-breakable witness} for $X$ is a vertex cut $(L,R)$ in $G$ of size $|L\cap R|\le k$ where $|L\cap X|,|R\cap X|>q$. If there is no $(q,k)$-breakable witness for $X$, then $X$ is \emph{$(q,k)$-unbreakable}. 
We say that $X$ has \emph{adhesion} $\sigma$ if, for every connected component $C$ in $G\setminus X$, the size of its neighborhood is at most $|N_{G}(C)|\le \sigma$. 
A $(q,k)$-unbreakable decomposition with adhesion $\sigma$ is a tree decomposition of $G$ where every bag is $(q,k)$-unbreakable and has adhesion at most $\sigma$. The quality of the decomposition is measured by how small $q$ and $\sigma$ are compared to $k$.

This paper presents the first close-to-linear time FPT algorithm for computing unbreakable decomposition, thereby removing the core bottleneck to close-to-linear time FPT algorithms for many graph cut problems.\footnote{As per convention in the graph algorithms literature, given a graph with $m$ edges, algorithms with \emph{near-linear}, \emph{almost-linear}, and \emph{close-to-linear} time have running time of $O(m \log^{O(1)}m)$, $O(m^{1+o(1)})$, and $O(m^{1+\epsilon})$ for any constant $\epsilon>0$, respectively.} Below, we survey the development and impact of unbreakable decomposition.

\paragraph{History.}

In their breakthrough FPT algorithm for \textsc{Minimum $p$-Way Cut}\footnote{Given a graph $G$, this problem asks if we can delete $k$ edges to disconnect $G$ into at least $p$ connected components. This problem is often called \textsc{Min $k$-Cut}. But we choose to preserve $k$ as a parameter in unbreakable decomposition.}, Kawarbayashi and Thorup~\cite{kawarabayashi2011minimum} introduced the edge cut version of \emph{$(q,k)$-breakable witnesses}. However, their algorithm for finding these witnesses is problem-specific. A year later, Chitnis et al.~\cite{chitnis2016designing} presented a general algorithm for finding $(q,k)$-breakable witnesses for both the edge cut and vertex cut versions, leading to FPT algorithms for many problems including the terminal version of \textsc{Minimum $p$-Way Cut} and \textsc{Unique Games}, which in turn generalizes \textsc{Subset Feedback Vertex Set} and \textsc{Odd Cycle Transversal}. Unfortunately, the top-down divide-and-conquer technique used in \cite{kawarabayashi2011minimum,chitnis2016designing} is ineffective in solving a prominent problem, namely, \textsc{Minimum Bisection}.

To address this, Cygan et al.~\cite{cygan2014minimum} introduced the concept of \emph{unbreakable decomposition}, which allows them to exploit $(q,k)$-unbreakable sets in a bottom-up manner using dynamic programming. Using this technique, they successfully developed an FPT algorithm for \textsc{Minimum Bisection}. Later, in \cite{cygan2020randomized}, Cygan et al.~significantly improved the construction in \cite{cygan2014minimum} by showing an FPT algorithm that constructs an unbreakable decomposition with optimal unbreakability and adhesion parameters. This improvement resulted in further applications.

Since then, unbreakable decomposition has become a core technique in various results within the field. For instance, Lokshtanov, Saurabh, and Surianarayanan~\cite{lokshtanov2022parameterized} provided a polynomial-time construction of unbreakable decomposition in the edge cut version and used it to settle a long line of work of \cite{gupta2018fpt,gupta2018faster,kawarabayashi2020nearly} by showing a $(1+\epsilon)$-approximation FPT algorithm for \textsc{Minimum $p$-Way Edge-Cut} when parameterized by $p$ (instead of the usual cut size). Other applications include model checking~\cite{pilipczuk2021algorithms,lokshtanov2018reducing,schirrmacher2024model}, connectivity oracles under vertex failures~\cite{pilipczuk2021algorithms}, and FPT algorithms for \textsc{Multiway Node Hubs}~\cite{saurabh202364}, \textsc{Judicious Partitions}~\cite{lokshtanov2019balanced}, \textsc{Fair Bisection}~\cite{inamdar2023parameterized}, and deletion to bounded degree graphs~\cite{agrawal2022fixed}.

\paragraph{Bottleneck.}

Unfortunately, all known algorithms for constructing unbreakable decompositions share an important drawback: their running time is far from linear in the size of the graph, in contrast to other key algorithmic techniques such as important separators \cite{marx2006parameterized}, flow augmentation \cite{kim2024flow}, treewidth reduction \cite{marx2013finding}, and branching from half-integral solutions \cite{iwata20180}. The fastest known algorithm in literature still takes at least $\Omega(mn^{2})$ time in a graph with $n$ vertices and $m$ edges \cite{cygan2014minimum}. See \Cref{table:decomp}. Unbreakable decomposition has thus become a common bottleneck for obtaining close-to-linear time FPT algorithms for many of its applications. Indeed, Cygan et al.~\cite{cygan2020randomized} has stated that whether unbreakable decomposition admits a near-linear time construction is ``an interesting and challenging open problem''.

This paper gives an affirmative answer to this question up to an arbitrarily small polynomial factor. Furthermore, our unbreakability and adhesion parameters are optimal up to constant factors, almost matching the optimal guarantees of \cite{cygan2020randomized}.

\begin{theorem}

\label{thm:main} For any $0<\epsilon\leq1$, there is a randomized algorithm that runs in time $2^{\OO(\frac{k}{\epsilon}\log\frac{k}{\epsilon})}m^{1+\epsilon}$ and computes with high probability a $(\OO(\frac{k}{\epsilon}),k)$-unbreakable decomposition with adhesion  $\OO(\frac{k}{\epsilon})$.

\end{theorem}

Our decomposition also has $O(\log n)$ depth and satisfies the \emph{subtree unbreakability property} (see \Cref{def:UnbreakableDecomp}), both of which are not satisfied by the decomposition with the best unbreakability parameters \cite{cygan2020randomized}. Both properties are useful in applications. For example, the performance of the connectivity oracle of \cite{pilipczuk2021algorithms} crucially relies on the subtree unbreakability property, and its space depends on the depth of the decomposition. \cite{inamdar2023parameterized} also requires a low-depth decomposition.

\renewcommand{\arraystretch}{1.4}
\begin{table}
\begin{tabular}{|c|c|c|c|c|>{\centering}p{0.1\textwidth}|>{\centering}p{0.1\textwidth}|}
\hline 
\textbf{Reference} & \textbf{Time} & \textbf{Unbreakability} & \textbf{Adhesion} & \textbf{Depth} & \textbf{Subtree?} & \textbf{Vertex or edge}\tabularnewline
\hline 
\hline 
\cite{cygan2014minimum} & \textcolor{red}{$2^{O(k^{2})}n^{2}m$} & $({\color{red}2^{O(k)}},k)$ & \textcolor{red}{$2^{O(k)}$} & \textcolor{red}{$n$} & Yes & Vertex\tabularnewline
\hline 
\cite{cygan2020randomized} & $2^{O(k\log k)}{\color{red}n^{O(1)}}$ & $(i,i)$ $\forall i\le k$ & $k$ & \textcolor{red}{$n$} & \textcolor{red}{No} & Vertex\tabularnewline
\hline 
\cite{lokshtanov2022parameterized} & \textcolor{red}{$n^{O(1)}$} & $((k+1)^{5},k)$ & $k$ & \textcolor{red}{$n$} & \textcolor{red}{No} & \textcolor{red}{Edge}\tabularnewline
\hline 
\cite{inamdar2023parameterized} & $2^{O(k\log k)}{\color{red}n^{O(1)}}$ & $(9k,k)$ & $8k$ & $O(\log n)$ & \textcolor{red}{No} & Vertex\tabularnewline
\hline 
\multirow{2}*{\textbf{Thm~\ref{thm:DetailedMain}}} & $2^{O(\frac{k}{\epsilon}\log\frac{k}{\epsilon})}m^{1+\epsilon}$ & $((2\left\lceil \frac{1}{\epsilon}\right\rceil +3)k,k)$ & $(2\left\lceil \frac{1}{\epsilon}\right\rceil +2)k$ & $O(\frac{k}{\epsilon}\log n)$ & Yes & Vertex \tabularnewline
\cline{2-7} 
~ & $2^{O(\frac{k}{\epsilon}\log\frac{k}{\epsilon})}m^{1+\epsilon}$ & $(O(k/\epsilon),k)$ & $O(k/\epsilon)$ & $O(\log n)$ & Yes & Vertex\tabularnewline
\hline
\end{tabular}

\caption{Known constructions of unbreakable  decomposition. The ``subtree?'' column indicates whether the decomposition satisfies the \emph{subtree unbreakability} property (see \Cref{def:UnbreakableDecomp}). The ``vertex or edge'' column indicates whether each bag is unbreakable or (weaker) edge-unbreakable.\label{table:decomp}}

\end{table}

As an almost immediate application of \Cref{thm:main}, we show the first close-to-linear time FPT algorithm for \textsc{Minimum $p$-Way Cut}.

\begin{theorem}
\label{thm:kcut} For any $0<\epsilon<1$, there is a $2^{\OO(\frac{k}{\epsilon}\log\frac{k}{\epsilon})}m^{1+\epsilon}$ time algorithm that decides if a graph $G$ has a \textsc{$p$-Way Cut} with at most $k$ cut edges.
\end{theorem}

Prior to our work, \cite{kawarabayashi2011minimum} showed an algorithm with a running time of $\OO(k^{k^{\OO(k)}}n^{2})$. This was improved by \cite{chitnis2016designing} to $\OO(2^{\OO(k^{2}\log p)}n^{2})$ and by \cite{cygan2020randomized} to $\OO(2^{\OO(k\log k)}n^{\OO(1)})$, where the dependency on $n^{\OO(1)}$ is at least $\Omega(mn^{2})$. 

We believe that this is only the first of many such applications.  Since there are numerous problems whose only solution so far is using unbreakable decomposition,~\Cref{thm:main} removes the main obstacle and opens up many exciting open problems listed in \Cref{sec:conclusion}.

\section{Technical Overview}

Given parameters $k$ and $\epsilon$, we define $\sigma=k+k\left\lceil \frac{1}{\epsilon}\right\rceil $ and $q=k +2\sigma$. Our strategy is to recursively call the following key subroutine: Given a graph $H$ and a set of \emph{boundary} vertices $B\subseteq V(H)$, satisfying $|B|\leq2\sigma$, find a superset $X\supseteq B$ such that

\begin{enumerate}

\item \label{prop:key1} $X$ is $(q,k)$-unbreakable,

\item \label{prop:key2} $X$ has adhesion $2\sigma$, i.e.,\textbf{ }$|N_{H}(C)|\leq2\sigma$ for each connected component $C$ in $H\setminus X$,

\item \label{prop:key3} $X$ is $\frac{1}{2}$-balanced, i.e., $|V(C)|\leq|V(H)|/2$ for each connected component $C$ in $H\setminus X$.

\end{enumerate}
Given $X$, we will then create a node $t$ in the tree decomposition, set the bag $\beta(t)=X$, and recurse on the graphs $C\cup N_{H}(C)$ with the boundary $N_{H}(C)$ for every connected component $C$ of $H\setminus X$. Each recursive call will then create a tree node with parent $t$.

By simply calling this subroutine on $G$ and initial $B=\emptyset$, 
we would obtain an $(q,k)$-unbreakable decomposition with adhesion $2\sigma$ and finish. The recursion depth is $O(\log n)$ by the $\frac{1}{2}$-balanced condition, so it suffices to show a close-to-linear time algorithm for this key subroutine.
In fact, we obtain the key subroutine only when $|B| \le \sigma$. We will later show how to reduce the case of $|B| > \sigma$ to the case when $|B| \le \sigma$ at the end of this overview.

Suppose $|B| \le \sigma$.
If we find a vertex set $X_{1}$ that is $(\sigma,k)$-unbreakable, has adhesion $\sigma$, and is $\frac{1}{2}$-balanced, then we can simply return $X:=X_{1}\cup B$, because $X$ must be $(\sigma+|B|,k)$-unbreakable where $\sigma+|B|\le q$, has adhesion $\sigma+|B|\le2\sigma$, and is $\frac{1}{2}$-balanced.

Our key technical contribution (\Cref{coro:UnbreakLowAdhesion}) is a fast algorithm for computing $X_{1}$, i.e., an unbreakable low-adhesion balanced set. We do this in two main steps. The first step computes an unbreakable \emph{balanced origin} $X_{0}$, defined below. The second step, called the \emph{reducing adhesion} step, will return the desired set $X_{1}\supseteq X_{0}$.

\paragraph{Balanced Origin (\Cref{sec:sampleset}).}

We say that a vertex set $X$ is a $\frac{1}{2}$-balanced $\sigma$-origin if every superset $X'\supseteq X$ with adhesion $\sigma$
must be $\frac{1}{2}$-balanced. We will compute a $\frac{1}{2}$-balanced $\sigma$-origin that is also $(k,k)$-unbreakable.

We first observe that we can obtain a $\frac{1}{2}$-balanced $\sigma$-origin $W$, with constant probability, simply by sampling a random vertex set of size $O(\sigma)$. This holds because, as shown in \cite{fm06}, with constant probability, $W$ is even a $(\frac{1}{2},\sigma)$-net which is a stronger notion. Let us assume that this event holds with certainty for simplicity. 

Now, if $W$ is also $(k,k)$-unbreakable, then we can return $X_0\gets W$ and be done. Else, there is a $(k,k)$-breakable witness $(L,R)$ for $W$. We can find such a witness simply by trying all partitions $(W_L,W_R)$ of $W$ and checking if there is a cut of size $k$ separating $(W_L,W_R)$. This takes $O(2^{|W|}mk)=2^{O(k/\epsilon)}m$ time. Assume w.l.o.g.~that $|L|\le|R|$. We update  $W\gets (W\setminus L)\cup(L\cap R)$. 
Our key observation is that $W$ remains a $\frac{1}{2}$-balanced $\sigma$-origin (because $|L|\le|V(H)|/2$) and the size of $W$ strictly decreases. Hence, this update can happen at most $O(\sigma)$ times before $W$ becomes $(k,k)$-unbreakable. Thus, we obtain a $(k,k)$-unbreakable $\frac{1}{2}$-balanced $\sigma$-origin $X_0$ in $2^{O(k/\epsilon)}m$ time.

\paragraph{Reducing Adhesion (\Cref{sec:reducing_adhesion}).}

Our \emph{reducing adhesion} subroutine (\Cref{thm:reducing_adhesion}) is such that, given any $(k,k)$-unbreakable set $X_0$, it finds a superset of $X_0$ that is $(\sigma ,k)$-unbreakable and has adhesion $\sigma$.

Given this, we can obtain the desired set $X_1$ by feeding a $(k,k)$-unbreakable $\frac{1}{2}$-balanced $\sigma$-origin $X_0$ as the input to the reducing adhesion subroutine. Then, the output $X_1$ must be $(\sigma,k)$-unbreakable and have adhesion $\sigma$. Also, $X_1$ must be $\frac{1}{2}$-balanced, because $X_1\supseteq X_0$ has adhesion $\sigma$ and $X_0$ is a $\frac{1}{2}$-balanced $\sigma$-origin.

The high-level scheme of the reducing adhesion subroutine is as follows.
We initialize $T \gets V(H)$ and $X_1 \gets X_0$ and keep updating $T$ and $X_1$, while maintaining the following three invariants:
 \begin{enumerate}
 \item \label{inv:1} $X_0 \subseteq X_1 \subseteq T$,
 \item \label{inv:2} $T$ has adhesion $\sigma$, and
 \item \label{inv:3} $X_1$ is $(\sigma,k)$-unbreakable.
 \end{enumerate} 
We stop once $X_1 = T$. Thus, $X_1$ will be $(\sigma,k)$-unbreakable and has adhesion $\sigma$ as desired.

 Our approach for updating $T$ and $X_1$ is centered around an object called
 a \emph{$(X_1,T,k')$-witness}, which is a vertex cut $(L,R)$ of size $|L\cap R|\le k'$ where $|L \cap T|>|L\cap R|$ and $X_1 \subseteq R$. This notion is similar to a $(k',k')$-breakable witness for $T$, but it requires $X_1 \subseteq R$ and omits $|R \cap T|>|L\cap R|$.  
 Given a $(X_1,T,k')$-witness $(L,R)$,  
 our algorithm updates $T$ by \emph{carving} $T$ along the witness: we set  $T \gets (T \setminus L) \cup (L \cap R)$. Before carving, a vertex $v \in L\setminus R$ is called \emph{$(X_1,T,k')$-carvable}.\footnote{The precise definition of carvable vertices (\Cref{def:Carvable}) requires a technical condition that the witness is ``connected'' in some sense. We omit this detail here.} 
 
 The carving operation maintains $X_1 \subseteq T$ because $X_1 \subseteq R$. Thus, Invariants \ref{inv:1} and \ref{inv:3} are maintained. In fact, $X_1$ remains $(k,k)$-unbreakable because $X_1$ has never been updated. We can also maintain Invariant \ref{inv:2} by making the $(X_1,T,k')$-witness \emph{lean} (see \Cref{def:LeanWitness}). The leanness is easy to ensure using max-flow computation; we omit this detail here. The point of the carving operation is that it removes $|(L \setminus R) \cap T| \geq 1$ carvable vertices from $T \setminus X_1$, progressing towards the goal of $T=X_1$. 

 Our key structural lemma  (\Cref{lemma:nocarvable}) says that if $X_1$ is $(q_1,k)$-unbreakable for any $q_1$ and 
 there is no $(X_1,T,q_1+k)$-carvable vertex in $T$,
 then $T$ is $(q_1+k,k)$-unbreakable. 
 This lemma suggests a natural algorithm: keep finding a $(X_1,T,2k)$-witness and carving $T$ along it until $T= X_1$ or no such witness is left. If the former happens, we are done. Otherwise, if the latter happens, the lemma implies that $T$ is $(2k,k)$-unbreakable, as $X_1$ is $(k,k)$-unbreakable. Since $T$ also has adhesion $\sigma$ by induction, we can set $X_1 \gets T$ and return $X_1$.

While the above approach is correct, it is too slow. Indeed, a single carving operation might reduce the size of $T \setminus X_1$ by only one. Thus, the process could take as large as $\Omega(n)$ iterations.

\paragraph{Reducing Adhesion Fast.}
To speed up, our key algorithmic tool is the \emph{disjoint-witness} algorithm (\Cref{coro:LeanWitCover}) that computes 
\begin{itemize}
\item A vertex set $Q$ that contains all $(X_{1},T,k')$-carvable vertices in $T$,
\item A collection ${\cal C}$ of disjoint $(X_{1},T,k')$-witnesses ${\cal C}$ where $(L\setminus R)\cap T\subseteq Q$ for all $(L,R)\in{\cal C}$ and $\sum_{(L,R)\in{\cal C}}|(L\setminus R)\cap T|=\Omega_{k'}(|Q|/\log n)$. 
\end{itemize}
Intuitively, the witnesses in ${\cal C}$ contains $\Omega_{k'}(1/\log n)$-fraction of all carvable vertices in $T$. The disjoint-witness algorithm takes $2^{O(k'\log k')}m^{1+o(1)}$ time and is based on the color coding technique in \cite{cygan2020randomized} combined with the single source min-cut threshold algorithm introduced by \cite{li2021approximate}, adapted to the vertex version by \cite{pettie2022optimal}. 

For simplicity, we first explain how to compute $X_{1}$ in $O_{k}(m^{1+o(1)}\sqrt{n})$ running time using a 2-level algorithm. Our final algorithm reduces the factor $\sqrt{n}$ to $n^{\epsilon}$ by having $\left\lceil \frac{1}{\epsilon}\right\rceil $ levels. 

Start by computing the set $Q$ of $(X_{1}, T, 2k)$-carvable vertices in $T$ and the collection ${\cal C}$ of disjoint $(X_{1},T,2k)$-witnesses. There are two cases depending on whether $|Q|\le\sqrt{n}$ or not.
If $|Q|>\sqrt{n}$, then we carve $T$ along all witness of ${\cal C}$. Note that the size of $T\setminus X_{1}$ must decrease by at least $\sum_{(L,R)\in{\cal C}}|(L\setminus R)\cap T|=\Omega_{k}(\sqrt{n}/\log n)$. So, this can happen at most $O_{k}(\sqrt{n}\log n)$ times. The total running time in this case is at most $O_{k}(m^{1+o(1)}\sqrt{n})$. 

Next, suppose $|Q|\le\sqrt{n}$. At a first glance, one may expect at most $\sqrt{n}$ further iterations because there are at most $\sqrt{n}$ many $(X_{1},T,2k)$-carvable vertices left.
Unfortunately, we cannot find a way to show this (and leave this as an open problem) since each carving operation might introduce new $(X_{1},T,2k)$-carvable vertices. This happens because the operation adds some new vertices into $T$. Even though $|T\setminus X_{1}|$ strictly decreases, there might be too many iterations because it is possible that $|T\setminus X_{1}|=\Omega(n)$ while $|Q|\le\sqrt{n}$. 

Our solution is to use the extension of the key structural lemma (\Cref{lemma:nocarvable}) above: for any set $Y$ where $X_{1}\subseteq Y\subseteq T$, if $X_{1}$ is $(q_{1},k)$-unbreakable and there is no $(X_{1},T,q_{1}+k)$-carvable vertex in $Y$, then $Y$ is $(q_{1}+k,k)$-unbreakable. 
Since $T\setminus Q$ has no $(X_{1},T,2k)$-carvable vertex, we can set $X_{1}\gets T\setminus Q$ (satisfying Invariant~\ref{inv:1}). The lemma implies that the new $X_1$ is $(2k,k)$-unbreakable (satisfying Invariant~\ref{inv:3}) and $|T\setminus X_{1}|=|Q|\le\sqrt{n}$. After this operation, the algorithm proceeds to the second \emph{level}.

In the second level, we keep finding a $(X_{1},T,3k)$-witness and carve $T$ along it, until $T=X_{1}$ or no such witness is left. At the beginning of this level,  $|T\setminus X_{1}|\le\sqrt{n}$, and so there are at most $\sqrt{n}$ further iterations. Once there is no $(X_{1},T,3k)$-witness, we return $X_{1}\gets T$. We conclude that $X_{1}$ is $(3k,k)$-unbreakable by the key structural lemma and $X_{1}$ has adhesion $\sigma$ by induction, as desired. 

We can extend the above 2-level algorithm to 3 levels simply by replacing the $\sqrt{n}$ threshold by two thresholds $n^{2/3}$ and $n^{1/3}$. By extending this approach to $\left\lceil \frac{1}{\epsilon}\right\rceil$ levels, we can obtain $X_{1}$ which is $(\left\lceil \frac{1}{\epsilon}\right\rceil k+k=\sigma,k)$-unbreakable and has adhesion $\sigma$ in time $O_{k}(m^{1+o(1)}n^{\epsilon})$ as desired.

\paragraph{When $|B|>\sigma$.}
We briefly discuss the case when $\sigma<|B|\le2\sigma$. 
Suppose that there exists a $(k,k)$-breakable witness $(L,R)$ for $B$. Let $X = (L\cap R)\cup B$. Since $|X|\le |B|+k \le 2\sigma+k=q$, $X$ is trivially $(q,k)$-unbreakable. Moreover, we can show that the adhesion of $X$ is at most $|B|-1$. That is, the size of the boundary strictly decreases, making progress towards the case where $|B|\le\sigma$. 
So, while we have that $X$ does satisfy Properties \ref{prop:key1} and \ref{prop:key2} but may not satisfy Property \ref{prop:key3}, this case may occur at most $\sigma$ times before $|B|\le\sigma$. This is why the depth of our decomposition has a factor of $\sigma$. In \Cref{sect:DepthReduction}, we further improve the depth to $O(\log n)$. 

The next case is when $B$ is $(k,k)$-unbreakable. If we apply the reducing adhesion subroutine to $B$, we get a set $X_2$ which is $(\sigma,k)$-unbreakable and has adhesion $\sigma$. This means that each connected component in $H\setminus X_2$ has a boundary size at most $\sigma$. Therefore, this reduces the problem to the previous case of  $|B|\le\sigma$ in close-to-linear FPT time.

\paragraph{Organization:}
\Cref{sec:preliminaries} contains basic definitions.
\Cref{sec:sampleset} describes the construction of {unbreakable balanced origins}. \Cref{sec:reducing_adhesion} describes our fast algorithm for reducing adhesion. Finally, using the tools from these two sections as outlined above, \Cref{sec:Decomp} shows the construction of the unbreakable decomposition. 
As an example of applications, \Cref{sec:app} shows the close-to-linear-time FPT algorithm for {\textsc{$p$-Way Cut}}. We list potential applications of our result and more open problems in \Cref{sec:conclusion}.

\section{Preliminaries}\label{sec:preliminaries}
Throughout this paper, we use $n$ to denote the number of vertices and $m$ to denote the number of edges in a graph. All graphs are undirected, unweighted, and connected unless otherwise stated, and hence we shall assume $m \ge n-1$. Given a graph $G$ and a subset of vertices $U \subseteq V(G)$, $G[U]$ denotes the subgraph induced on the set of vertices $U$ and $E_G(U)$ denotes the set of edges in $E(G)$ whose both endpoints are in $U$, i.e. the set of edges in $G[U]$. The set of neighbours of a vertex $v$ in $G$ is denoted by $N_G(v)$. We denote by $N_G(U):= \bigcup_{v \in U} N_G(v) \setminus U$ the set of neighbours of $U$. We omit the subscripts when the graph is clear from the context. 

\paragraph{Vertex Cuts.} 
A \emph{vertex cut} $(L,R)$ is such that $L \cup R = V(G)$, $L\setminus R$ and $R\setminus L$ are not empty, and there is no edge between $L\setminus R$ and $R\setminus L$. The \emph{size} of the vertex cut $(L,R)$ is $|L\cap R|$. We emphasize that, we will view $(L,R)$ as an \emph{ordered pair}, because in some definitions (e.g.~\Cref{def:DisjointCuts,def:LeanWitness}), the order of $(L,R)$ does matter.
Throughout \Cref{sec:sampleset,sec:singlesource,sec:reducing_adhesion,sec:Decomp}, we only consider vertex cuts (not edge cuts), so we usually write cuts as an abbreviation of vertex cuts.

\paragraph{Vertex-Capacitated Graphs and Mincuts.}

We use $\wtilde{G}$ to denote \emph{capacitated graphs} with positive integral vertex capacity function $\rho: V(\wtilde{G}) \rightarrow \mathbb{Z}^+$.
For an arbitrary vertex set $A\subseteq V(\wtilde{G})$, we define $\rho(A) = \sum_{v\in A}\rho(v)$. Given two disjoint vertex sets $A,B \in V(\wtilde{G})$ such that there is no edge connecting $A$ and $B$, an $A$-$B$ cut is a cut $(L,R)$ with $A\subseteq L\setminus R$ and $B\subseteq R\setminus L$, and an $A$-$B$ mincut is an $A$-$B$ cut that minimizes the (capacitated) cut size $\rho(L\cap R)$. We use $\lambda_{\wtilde{G}}(A,B)$ to denote the (capacitated) size of an $A$-$B$ mincut. To avoid clutter, when $A = \{a\}$ (resp. $B = \{b\}$) is a singleton vertex, we replace $\{a\}$ with $a$ (resp. replace $\{b\}$ with $b$).

\paragraph{Balance, Adhesion and Unbreakability.}

For any vertex set $X\subseteq V(G)$, $X$ is \emph{$\alpha$-balanced} if each connected component $C$ of $G\setminus X$ has size $|C|\le \alpha n$. The \emph{adhesion} $\sigma_G(X)$ of $X$ in $G$, or simply the adhesion of $X$, is the maximum, over connected components $C$ of $G\setminus X$, of the quantity $|N_G(C)|$.

\begin{definition}[Unbreakablility]
A vertex set $X \subseteq V(G)$ is $(q,k)$-\emph{unbreakable in $G$} if every vertex cut $(L,R)$ of size at most $k$ satisfies $|L\cap X|\leq q$ or $|R\cap X|\leq q$. A $(q,k)$-\emph{breakable witness of $X$ in $G$} is a vertex cut $(L,R)$ of $G$ of size as most $k$ satisfying $|L\cap X| >  q$ and $|R\cap X|> q$.
\end{definition}
By definition, $X$ is $(q,k)$-unbreakable in $G$ iff there is no $(q,k)$-breakable witness of $X$ in $G$. Note that any set $X$ of size at most $q$ is vacuously $(q,k)$-unbreakable.

The core technical contribution of this paper (\Cref{coro:UnbreakLowAdhesion}) is a close-to-linear time algorithm for computing a vertex set that is simultaneously $1/2$-balanced, unbreakable, and has low-adhesion.

\paragraph{Unbreakable Decomposition.} Next, we define the key object of this paper.
 
\begin{definition}[Tree Decomposition]
\label{def:TreeDecomp}
A \emph{tree decomposition} of a graph $G$ is a pair $(T,\beta)$, where $T$ is a tree and $\beta: V(T) \rightarrow 2^{V(G)}$ is a mapping that assigns to every tree node $t$ a subset $\beta(t)\subseteq V(G)$, called a \emph{bag}. Furthermore, $(T,\beta)$ satisfies the following.
\begin{itemize}
    \item For each vertex $v\in V(G)$, the set $\{t \mid v\in \beta(t)\}$ induces a connected subtree of $T$.
    \item For each edge $\{u,v\}\in E(G)$, there is a tree node $t$ where $u,v\in \beta(t)$.
\end{itemize}
\end{definition}

A \emph{rooted} tree decomposition is a tree decompositon $(T, \beta)$ together with a designated root node $r \in V(T)$. For any node $t\in V(T)$ with parent $t'$, the adhesion of a tree node $t$ is $\sigma(t) = \beta(t) \cap \beta(t')$. We define $\sigma(r) = \emptyset$ for the root $r$. The \emph{adhesion} of $T$ is $|\max_{t \in T}\sigma(t)|$. 
For every $t\in V(T)$, we also define the sets 
\[
\gamma(t) = \bigcup_{\text{descendants $s$ of $t$}} \beta(s)\text{ and }G_t = G[\gamma(t)] - E_G(\sigma(t)).
\]

\begin{definition}[Unbreakable Decomposition]\label{def:UnbreakableDecomp}A $(q,k)$-unbreakable decomposition of $G$ is a rooted tree decomposition $(T,\beta)$ where each bag $\beta(t)$ is $(q,k)$-unbreakable in $G$. The decomposition admits the stronger \emph{subtree unbreakability} property if each bag $\beta(t)$ is $(q,k)$-unbreakable in $G_{t}$.
\end{definition}

The main goal of this paper (\Cref{sec:Decomp}) is a fast algorithm for computing an unbreakable decomposition with  subtree unbreakability property and  small adhesion.

For each tree node $t$, let $\alpha(t) = \gamma(t)\setminus\sigma(t)$.
We say that a rooted tree decomposition $(T, \beta)$ is \emph{compact} if for every node $t \in V(T)$ for which $\sigma(t) \neq \emptyset$, $G[\alpha(t)]$ is connected and $N_G(\alpha(t)) = \sigma(t)$. This property is handy for performing dynamic programming on tree decomposition, and we will exploit it in \Cref{sec:app}.

\paragraph{Single Source Vertex Mincuts.} We will use an algorithm for computing single-source vertex min-cuts. We start by defining the notion of \emph{disjoint} cuts.

\begin{definition}[Disjoint Cuts]
Let ${\cal C}$ be a collection of cuts. The cuts in ${\cal C}$ are disjoint if for each pair of different cuts $(L,R),(L',R')\in{\cal C}$, $L\setminus R$ and $L'$ are disjoint.
\label{def:DisjointCuts}
\end{definition}
A cut collection ${\cal C}$ is a set of cuts $(L,R)$. 
Next, we define a \emph{mincut cover}, which is the output of our single source vertex mincuts subroutine.

\begin{definition}[Mincut Covers]
\label{def:MincutCovers}
Consider a capacitated graph $\wtilde{G}$ with a source vertex $s$ and sink vertices $T$ such that $\{s\}\cup T$ is an independent set. A \emph{mincut cover} ${\cal K}$ with respect to $s$ and $T$ in $\wtilde{G}$ is a set of cut collections ${\cal C}$, which satisfies the following.
\begin{enumerate}
\item\label{prop:Cover1} For each collection ${\cal C}\in{\cal K}$ and cut $(L,R)\in{\cal C}$, $(L,R)$ is a $t$-$s$ mincut for some sink $t\in T$. 
\item\label{prop:Cover2} For each sink $t\in T$, there exists a cut $(L,R)$ in some collection ${\cal C}\in{\cal K}$ such that $t\in L\setminus R$.
\item\label{prop:Cover3} Each collection ${\cal C}$ is a set of disjoint cuts.
\end{enumerate}
The \emph{width} of a mincut cover ${\cal K}$ is the number of collections ${\cal C}$ in ${\cal K}$. To avoid clutter, we also use $(L,R)\in{\cal K}$ to denote a cut $(L,R)\in{\cal C}$ for some collection ${\cal C}\in{\cal K}$.
\end{definition}

Roughly speaking, the single source vertex mincuts subroutine receives a capacitated graph $\wtilde{G}$ with one source $s$, a set $T$ of several sinks, and a parameter $k$, and outputs a small-width mincut cover ${\cal K}$ with respect to $s$ and $T$. In other words, we can obtain $t$-$s$ mincuts for all $t\in T$, and these mincuts can be partitioned into a small number of collections of disjoint cuts. We defer the proof of this result to~\Cref{sec:singlesource}.

\begin{restatable}{theorem}{thmssmc}
Consider an $m$-edge capacitated graph $\wtilde{G}$ with vertex capacity function $\rho$, a parameter $k$, a single source vertex $s$ and sink vertices $T$ satisfying that $\{s\}\cup T$ is an independent set and each source/sink vertex has capacity $\infty$.
Let $T^{*}$ be the set of sink vertices $t$ with $\lambda_{\wtilde{G}}(t,s)\leq k$. There is a randomized algorithm that, with high probability, computes a mincut cover ${\cal K}$ with respect to $s$ and $T^{*}$ which has width $O(k\log^{3}n)$.
The running time is $\OO(km^{1+o(1)})$.
\label{thm:SSMC}
\end{restatable}

\section{Unbreakable Balanced Origins from Nets}
\label{sec:sampleset}

Let us define a \emph{balanced origin}, a vertex set where every low-adhesion superset must be balanced. 

\begin{definition}[Balanced Origin]
Given a graph $G$, an $\alpha$-balanced $\sigma$-origin is a vertex set $X \subseteq V(G)$ such that, for any vertex set $X'\supseteq X$ with adhesion $\sigma$, $X'$ must be $\alpha$-balanced.

\end{definition}

The goal of this section is to compute an unbreakable balanced origin.

\begin{lemma}
\label{lemma:SepBySampleSet}
Given an $m$-edge graph $G$ with parameters $k,\sigma$ where $k\leq \sigma$, there is a randomized algorithm that computes a set $X$ which is always $(k,k)$-unbreakable such that with constant probability, $X$ is a  $\frac{1}{2}$-balanced $\sigma$-origin. The running time is $2^{\OO(\sigma)}m$.

\end{lemma}

The high-level idea of the algorithm for \Cref{lemma:SepBySampleSet} is to first sample a \emph{net} (\Cref{lemma:Feige}) as the initial set and keep updating the set as long as there exists a breakable witness until it becomes unbreakable in a straightforward manner using \Cref{lemma:CheckUnbreakable}. The properties of nets will then allow us to show that the final set is a $\frac{1}{2}$-balanced $\sigma$-origin.

We now formally define the notion of $(\alpha, \sigma)$-nets. This concept was introduced by~\cite{fm06} and is closely related to the notion of \emph{detection sets}%
~\cite{kleinberg2004detecting}. %

\begin{definition}[$(\alpha,\sigma)$-nets]
\label{def:nets}
A set $W$ of vertices in a graph $G$ is an $(\alpha,\sigma)$-net if for every set of vertices $S$ in $G$ of size at most $\sigma$, and for every connected component $D$ of $G\setminus S$,
\begin{enumerate}
\item If $|D|\geq \alpha n$, then $D$ has at least one vertex from $W$.
\item If $|D|\leq (1-\alpha)n-|S|$, then the set $V(G)\setminus (D\cup S)$ has at least one vertex from $W$.
\end{enumerate}
\end{definition}

We remark that we will only use the first property in our proof of \Cref{lemma:SepBySampleSet} and in the whole paper.

\begin{lemma}[Corollary 3.6 of~\cite{fm06}]
\label{lemma:Feige}
Given a graph $G$, there exists an absolute constant $c$ such that any random subset $W \subseteq V(G)$ of size $c\frac{\sigma}{\alpha}\log \frac{1}{\alpha}$ is an $(\alpha, \sigma)$-net with constant probability.

\end{lemma}

\begin{lemma}
Given an $m$-edge graph $G$ with a set $W\subseteq V(G)$ and parameters $q,k$, there is a deterministic algorithm that either
\begin{itemize}
\item certifies that $W$ is $(q,k)$-unbreakable in $G$, or
\item outputs a $(q, k)$-breakable witness $(L,R)$ of $W$ in $G$.
\end{itemize}
The running time is $\OO(2^{|W|}k m)$.
\label{lemma:CheckUnbreakable}
\end{lemma}
\begin{proof}
For each $W_{L}\subseteq W$ and $W_{R} = W\setminus W_{L}$ s.t. $|W_{L}|,|W_{R}| > q$, we compute a vertex mincut $(L,R)$ that separates $W_{L}$ from $W_{R}$ in $G$. %
If the cut size is at most $k$, then $(L,R)$ is a $(q,k)$-breakable witness of $W$ in $G$ by definition, and we terminate the algorithm with output $(L,R)$. %
If there is no such $(L,R)$ after checking all $(W_{L}, W_{R})$, it must be the case that there is no $(q, k)$-breakable witness of $W$ in $G$, so $W$ is $(q,k)$-unbreakable in $G$ by definition.

Now we analyze the running time. The number of partitions $(W_{L}, W_{R})$ of $W$ is at most $2^{|W|}$. For each of them, we invoke a maxflow algorithm\footnote{We can use the classic Ford-Fulkerson algorithm because we only want to obtain a mincut of size at most $k$ or decide such mincut does not exist.} on $G$, which takes $\OO(km)$ time. Hence the total running time is $\OO(2^{|W|}k m)$.
\end{proof}

We are now ready to prove \Cref{lemma:SepBySampleSet}.

\begin{proof}[Proof of \Cref{lemma:SepBySampleSet}]
Set $\alpha = \frac{1}{2}$. Using \Cref{lemma:Feige}, we first compute an $(\alpha, \sigma)$-net $W$ of $G$ with size $|W| = O(\frac{\sigma}{\alpha}\log(\frac{1}{\alpha})) = O(\sigma)$. We initialize $X^{(0)}$ to be $W$, and then we will update this set iteratively as follows. Let $X^{(i)}$ be the set of vertices right after the $i$-th iteration. For all $i\geq 0$, we will maintain the invariant that for each vertex $v\in W\setminus X^{(i)}$, the connected component in $G\setminus X^{(i)}$ containing $v$ has size at most $n/2$. Initially, the invariant vacuously holds because $X^{(0)} = W$ and so $W\setminus X^{(0)} = \emptyset$.

At the $i$-th iteration, we use \Cref{lemma:CheckUnbreakable} to check whether $X^{(i-1)}$ is $(k,k)$-unbreakable in $G$ or not. If it is indeed $(k,k)$ unbreakable, we set $X = X^{(i-1)}$ and terminate the whole algorithm with output $X$. Otherwise, \Cref{lemma:CheckUnbreakable} will return a $(k,k)$-breakable witness $(L,R)$ of $X^{(i-1)}$ in $G$. Without loss of generality, we assume $L$ is the smaller side, i.e. $|L|\leq |R|$. Then we update $X^{(i)} = (X^{(i-1)}\setminus L)\cup (L\cap R)$. Namely, we first remove the part in the smaller side $L$, and then add $L\cap R$.

Now, we show that the invariant holds for $X^{(i)}$, assuming that $X^{(i-1)}$ already satisfies the invariant. Consider a vertex $v\in W\setminus X^{(i)}$. 
First, if $v\in L\setminus R$, the connected component $D_{i}$ of $G\setminus X^{(i)}$ containing $v$ satisfies $D_{i}\subseteq L\setminus R$, so $|D_{i}|\leq |L\setminus R|\leq n/2$. From now, we consider the case that $v\in R$. In fact, we must have $v\in R\setminus L$ in this case because $L\cap R\subseteq X^{(i)}$ but $v\in W\setminus X^{(i)}$. Let $D_{i}$ be the connected component of $G\setminus X^{(i)}$ containing $v$. Note that $D_{i}\subseteq R\setminus L$ because $v\in R\setminus L$ and $L\cap R\subseteq X^{(i)}$. Furthermore, by the update rule, $X^{(i-1)}\cap (R\setminus L) = X^{(i)}\cap (R\setminus L)$, so $D_{i}$ is disjoint from $X^{(i-1)}$, which means $D_{i}$ is inside a connected component $D_{i-1}$ of $G\setminus X^{(i-1)}$ and in particular $|D_{i}|\leq |D_{i-1}|$. Again by $X^{(i-1)}\cap (R\setminus L) = X^{(i)}\cap (R\setminus L)$, we know $v\in W\setminus X^{(i-1)}$ combining $v\in R\setminus L$ and $v\in W\setminus X^{(i)}$. Therefore, the invariant of $X^{(i-1)}$ gives that $|D_{i-1}|\leq n/2$, which implies $|D_{i}|\leq |D_{i-1}|\leq n/2$ as desired.

At the end of the algorithm (we will discuss why it must end in the running time analysis), we obtain a $(k,k)$-unbreakable set $X$ such that for each vertex $v\in W\setminus X$, the connected component in $G\setminus X$ containing $v$ has size at most $n/2$. Note that the unbreakability of $X$ and the invariant always hold regardless of the success probability of \Cref{lemma:Feige}.

Next we show that $X$ is a $\frac{1}{2}$-balanced $\sigma$-origin conditioned on the success of \Cref{lemma:Feige}. Since~\Cref{lemma:Feige} is satisfied with constant probability, it would then follow that $X$ is a $\frac{1}{2}$-balanced $\sigma$-origin with constant probability. Let $X'\supseteq X$ be a vertex set with adhesion $\sigma$. Consider a connected component $D'$ of $G\setminus X'$, and assume for contradiction that $|D'|>n/2$. We know that $|N_{G}(D')|\leq \sigma$ since the adhesion of $X'$, $\sigma(X')$, is at most $\sigma$. Then since $|D'|>n/2$, $|N_{G}(D')|\leq \sigma$ and by the definition of $(\alpha,\sigma)$-nets (property $1$ of~\Cref{def:nets}), $W$ has at least one vertex from $D'$, say $v$. Then $v\in W\setminus X$. However, by the invariant, the connected component $D$ of $G\setminus X$ containing $v$ must have size $|D|\leq n/2$, which implies that $|D'|\leq |D|\leq n/2$, a contradiction. 

\medskip

\noindent{\underline{Running Time.}} The bottleneck is applying \Cref{lemma:CheckUnbreakable} in each iteration. The number of iterations is at most $|W|$ since the initial set is $X^{(0)} = W$ and each update to the set decreases its size by at least $1$. Therefore, the total running time is $\OO(2^{|W|}\cdot|W|\cdot km) = 2^{\OO(\sigma)}m$ (since $\sigma\geq k$).
\end{proof}

\section{Reducing Adhesion of Unbreakable Sets}
\label{sec:reducing_adhesion}

This subsection forms the main technical component of our unbreakable decomposition. Roughly speaking, \Cref{thm:reducing_adhesion} shows, given an unbreakable set $X_{0}$, in almost linear time, we can expand it to another set $X$ which is appropriately unbreakable and has small adhesion.

\begin{lemma}\label{thm:reducing_adhesion}
Given an $m$-edge graph $G$ with parameters $0<\epsilon\leq 1$, $k\geq 1$ and $q\geq k$, and an initial set $X_0 \subseteq V(G)$ such that $X_0$ is $(q,k)$-unbreakable, there is an algorithm that computes with high probability that a set $X \supseteq X_0$ such that $X$ is $(q + k \lceil \frac{1}{\epsilon} \rceil, k)$ unbreakable and has adhesion at most $\sigma(X) = q + k \lceil \frac{1}{\epsilon} \rceil$. The running time is $\exp(\OO((q+\frac{k}{\epsilon})\log(q+\frac{k}{\epsilon}))) m^{1+\epsilon+o(1)}$.
\end{lemma}

Before proving \Cref{thm:reducing_adhesion}, we first give \Cref{coro:UnbreakLowAdhesion} which may be of independent interest. Roughly speaking, \Cref{coro:UnbreakLowAdhesion} is a simple corollary of \Cref{lemma:SepBySampleSet} and \Cref{thm:reducing_adhesion}, which says we can compute a unbreakable balanced vertex set with low adhesion efficiently.

\begin{theorem}[Unbreakable Balanced Low-Adhesion Sets]
\label{coro:UnbreakLowAdhesion}
Given an $m$-edge graph $G$ with parameters $0<\epsilon\leq 1$ and $k\geq 1$, there is a randomized algorithm that with high probability computes a set $X\subseteq V(G)$ such that $X$ is $(\lceil 1/\epsilon\rceil k + k, k)$-unbreakable, $1/2$-balanced and has adhesion $\lceil 1/\epsilon \rceil k + k$.
The running time is $\exp(\OO(\frac{k}{\epsilon}\log\frac{k}{\epsilon})) m^{1+\epsilon+o(1)}$. 
\end{theorem}
\begin{proof}
We do the following $O(\log n)$ times. First, let $X_{0}$ be an initial set by applying \Cref{lemma:SepBySampleSet} on $H$ with parameters $k$ and $\sigma = \lceil 1/\epsilon\rceil k + k$. Second, compute $X$ by applying \Cref{thm:reducing_adhesion} on graph $H$ with $X_{0}$ as the initial set and parameters $\epsilon,k,q=k$. If $X$ is $1/2$-balanced, we terminate the whole algorithm, otherwise proceed to the next iteration.

To see the correctness, consider each of the $O(\log n)$ iterations. \Cref{lemma:SepBySampleSet} guarantees that $X_{0}$ is always $(k,k)$-unbreakable and with constant probability, $X_{0}$ is a $1/2$-balanced $\sigma$-origin. Therefore, \Cref{thm:reducing_adhesion} guarantees that, with high probability, $X\supseteq X_{0}$ is $(\lceil 1/\epsilon\rceil k + k)$-unbreakable and has adhesion at most $\sigma = \lceil 1/\epsilon\rceil k + k$. Furthermore, because $X_{0}$ is a $1/2$-balanced $\sigma$-origin with constant probability, we have $X$ is $1/2$-balanced with constant probability. It follows that with high probability, at least one of the $O(\log n)$ iterations will give a $1/2$-balanced $X$.
\end{proof}

In the following subsections, we will prove \Cref{thm:reducing_adhesion}.

\subsection{Witnesses and Carvable Vertices}

We will introduce concepts around \emph{witnesses} and \emph{carvable vertices}, and then show several useful observations that eventually lead to the final algorithm for~\Cref{thm:reducing_adhesion}. 
\begin{definition}
Given a vertex set $T\subseteq V(G)$, the \emph{torso} $H_{T}$ of $T$ in $G$ is a graph with $V(H_{T})=T$ and $E(H_{T})=\{\{u,v\}\mid\{u,v\}\in G$ or $u,v\in N(D)$ for some connected component $D$ of $G\setminus T\}$. 
\end{definition}

Observe that $H_{T}$ has at most $|E(G)|(\sigma(T))^{2}$ edges where $\sigma(T)$ is the adhesion of $T$. We will refer to $H_{T}$ as $H$ when $T$ is clear from the context. This notion of the torso graph $H_T$ was  introduced by~\cite{cygan2020randomized} for their color coding step, and we will exploit it for our color coding step as well.
\begin{definition}
[Witnesses and Carvable Vertices]\label{def:Carvable} Given a vertex set $T$ and a set of vertices $X\subseteq T$, we say that a cut $(L,R)$ in $G$ is an $(X,T,k')$-\emph{witness} if 
\begin{enumerate}
\item\label{prop:wit1} $|L\cap R|\leq k'$, 
\item\label{prop:wit2} $|L\cap T|>|L\cap R|$, and 
\item\label{prop:wit3} $X\subseteq R$. 
\end{enumerate}
We say that a $(X,T,k')$-witness $(L,R)$ is \emph{connected} if $H_{T}[(L\setminus R)\cap T]$ is connected. A vertex $v$ is \emph{$(X,T,k')$-carvable} if there exists a connected $(X,T,k')$-witness where $v\in L\setminus R$. 
\end{definition}

This definition of carvable vertices is similar to that in~\cite{korhonen2024minor}. The two differences are that instead of the condition $|L \cap T| > |L \cap R|$,~\cite{korhonen2024minor} require $|(L \setminus R) \cap T| > \alpha$ for some $\alpha \gg k'$, and they require that $L \setminus R$ is connected in $G$, instead of requiring that $(L \setminus R) \cap T$ is connected in $H_T[(L \setminus R) \cap T]$.

The following structural lemma is crucial. It says that, for any vertex set $T$, if $X\subseteq T$ is unbreakable and we remove all carvable vertices from $T$ to obtain the set $Y$, then $Y$ is also unbreakable. 

\begin{lemma}
\label{lemma:nocarvable}For some arbitrary $q,k$ where $q\geq k$ and a set $Y$ such that $X\subseteq Y\subseteq T$, if $X$ is $(q,k)$-unbreakable and there is no $(X,T,q+k)$-carvable vertex in $Y$, then $Y$ is $(q+k,k)$-unbreakable. 
\end{lemma}

\begin{proof}
Suppose for contradiction that $Y$ is not $(q+k,k)$ unbreakable. Then, there exists a vertex cut $(L_{0},R_{0})$ of size $|L_{0}\cap R_{0}|\leq k$ where $|L_{0}\cap Y|>q+k$ and $|R_{0}\cap Y|>q+k$. Since $X$ is $(q,k)$-unbreakable, either $|L_{0}\cap X|\leq q$ or $|R_{0}\cap X|\leq q$. Without loss of generality, let us assume the former.

Observe that the vertex cut $(L_{0},R_{1}=R_{0}\cup(L_{0}\cap X))$ is a $(X,Y,q+k)$-witness. Let us verify the three conditions. First, the cut size is $|L_{0}\cap R_{1}|=|(L_{0}\cap R_{0})\cup(L_{0}\cap X)|\le q+k$. Second, we have $|L_{0}\cap Y|>q+k\ge|L_{0}\cap R_{1}|$. Last, we have $X\subseteq R_{0}\cup(L_{0}\cap X)\subseteq R_{1}$ by construction. 

Next, we introduce a new notion within this proof. A cut $(L,R)$ is a \emph{candidate-witness} if
\begin{enumerate}
\item $|L\cap R|\leq q+k$, 
\item $|(L\setminus R)\cap Y|>|(L\cap R)\setminus T|$, and 
\item $X\subseteq R$. 
\end{enumerate}
The notion of candidate-witness is  ``between'' $(X,Y,q+k)$-witness and $(X,T,q+k)$-witness. More formally, any $(X,Y,q+k)$-witness is a candidate-witness, because for any cut $(L,R)$ we have $|L\cap Y|>|L\cap R|$ iff $|(L\setminus R)\cap Y|>|(L\cap R)\setminus Y|$, and additionally $|(L\cap R)\setminus Y|\ge|(L\cap R)\setminus T|$ because $Y\subseteq T$. Similarly, any candidate-witness is a $(X,T,q+k)$-witness, because  $|(L\setminus R)\cap T|\ge|(L\setminus R)\cap Y|>|(L\cap R)\setminus T|$ which is equivalent to $|L\cap T|>|L\cap R|$.

Next, let $(L^{*},R^{*})$ be a candidate-witness where $|(L^{*}\setminus R^{*})\cap T|$ is minimized. Note that $(L^{*},R^{*})$ is well-defined by the existence of the $(X,Y,q+k)$-witness $(L_{0},R_{1})$.
\begin{claim}
$H_{T}[(L^{*}\setminus R^{*})\cap T]$ is connected.\label{lemma:Candidate} 
\end{claim}

Given this claim, we obtain a contradiction as follows. As observed above, $(L^{*},R^{*})$ is a $(X,T,q+k)$-witness, and by \Cref{lemma:Candidate}, it is a connected $(X,T,q+k)$-witness. Since $|(L^{*}\setminus R^{*})\cap Y|>0$, there exists a vertex $v\in(L^{*}\setminus R^{*})\cap Y$. So, $v$ is a $(X,T,q+k)$-carvable vertex in $Y$, which is a contradiction. It remains to prove \Cref{lemma:Candidate}. The proof of this claim is similar to the proof of Claim 3.8 in~\cite{cygan2020randomized}.
\begin{proof}
[Proof of \Cref{lemma:Candidate}]Let $\mathcal{C}=\{C_{1},C_{2},...,C_{q}\}$ denote the connected components of $G[L^{*}\setminus R^{*}]$. We first assume without loss of generality that each component $C_{i}$ intersects $T$. Otherwise, let $\Gamma$ denote the union of all components $C_{i}$ that are disjoint from $T$. Note that $\Gamma$ is also disjoint from $Y\subseteq T$. Observe that the cut $(L',R')$ with $L'=L^{*}\setminus\Gamma$ and $R'=R^{*}\cup\Gamma$ is still a $(X,Y,q+k)$-witness where $|(L'\setminus R')\cap T|=|(L^{*}\setminus R^{*})\cap T|$ is still minimum. So we can work with $(L',R')$ instead.

Assume $H_{T}[(L^{*}\setminus R^{*})\cap T]$ is not connected. We will reach contradiction by showing another candidate witness $(\hat{L},\hat{R})$ where $|(\hat{L}\setminus\hat{R})\cap T|<|(L^{*}\setminus R^{*})\cap T|$. 

Let $D_{1},D_{2}$ be a partition of $(L^{*}\setminus R^{*})\cap T$ such that there is no edge connecting $D_{1}$ and $D_{2}$ in $H_{T}[(L^{*}\setminus R^{*})\cap T]$. Note that each $C_{i}\cap T$ induces a connected subgraph of $H_{T}$, so we can write $D_{1}$ and $D_{2}$ in the form $D_{1}=\bigcup_{C_{i}\in{\cal C}_{1}}C_{i}\cap T$ and $D_{2}=\bigcup_{C_{i}\in{\cal C}_{2}}C_{i}\cap T$, where ${\cal C}_{1}$ and ${\cal C}_{2}$ partition $\mathcal{C}=\{C_{1},C_{2},...,C_{q}\}$.

The key observation is that, for any $C_{i_{1}}\in{\cal C}_{1}$ and $C_{i_{2}}\in{\cal C}_{2}$, we must have $N(C_{i_{1}})\cap N(C_{i_{2}})\subseteq T$\footnote{We emphasize that $N(\cdot)$ denote neighbors \underline{in $G$}.}. Otherwise, by the definition of $H_{T}$, the existence of a vertex in $(N(C_{i_{1}})\cap N(C_{i_{2}}))\setminus T$ implies that there is an edge in $E(H_{T})$ connecting some terminal $t_{1}\in C_{i_{1}}\cap T\subseteq D_{1}$ and some $t_{2}\in C_{i_{2}}\cap T\subseteq D_{2}$, contradicting the fact that there is no edge between $D_{1}$ and $D_{2}$ in $H_{T}[(L^{*}\setminus R^{*})\cap T]$.

Let $\wtilde C_{1}=\bigcup_{C_{i}\in{\cal C}_{1}}C_{i}$ and $\wtilde C_{2}=\bigcup_{C_{i}\in{\cal C}_{2}}C_{i}$. From this key observation, we further have $N(\wtilde C_{1})\cap N(\wtilde C_{2})\subseteq T$, which implies 
\[
(N(\wtilde C_{1})\setminus T)\cap(N(\wtilde C_{2})\setminus T)=\emptyset.
\]
Since $N(\wtilde C_{1})\cup N(\wtilde C_{2})=N(L^{*}\setminus R^{*})\subseteq L^{*}\cap R^{*}$ (since $\wtilde C_{1}\cup\wtilde C_{2}=L^{*}\setminus R^{*}$), we also have 
\[
(N(\wtilde C_{1})\setminus T)\cup(N(\wtilde C_{2})\setminus T)\subseteq(L^{*}\cap R^{*})\setminus T.
\]
Therefore, we have $|N(\wtilde C_{1})\setminus T|+|N(\wtilde C_{2})\setminus T|\leq|(L^{*}\cap R^{*})\setminus T|$. Also, $|\wtilde C_{1}\cap Y|+|\wtilde C_{2}\cap Y|=|(L^{*}\setminus R^{*})\cap Y|>|(L^{*}\cap R^{*})\setminus T|$. Combining both, we have either $|\wtilde C_{1}\cap Y|>|N(\wtilde C_{1})\setminus T|$ or $|\wtilde C_{2}\cap Y|>|N(\wtilde C_{2})\setminus T|$. Thus either one of the cuts $(\wtilde C_{1}\cup N(\wtilde C_{1}),V(G)\setminus\wtilde C_{1})$ or $(\wtilde C_{2}\cup N(\wtilde C_{2}),V(G)\setminus\wtilde C_{2})$ is a candidate witness. However, since both $\wtilde{C}_{1}$ and $\wtilde{C}_{2}$ intersect $T$, we have $|\wtilde C_{1}\cap T|,|\wtilde C_{2}\cap T|<|(L^{*}\setminus R^{*})\cap T|$, which contradicts the minimality of $(L^{*},R^{*})$.
\end{proof}
\end{proof}

\subsection{Carve Terminals with Lean Witnesses}

In this section, we analyze the basic building block of our algorithm for~\Cref{thm:reducing_adhesion}. We define the notion of \emph{lean witness}, describe how to carve a terminal set given a lean witness, and analyze how the adhesion and size of a terminal set changes. The notion of \emph{lean witness} in this paper essentially serves the same purpose as the term single bag lean witness introduced by~\cite{cygan2020randomized}.

\begin{definition}
\label{def:LeanWitness}
An $(X,T,k')$-witness $(L^{*},R^{*})$ is \emph{lean} if there exists a set of $|L^{*}\cap R^{*}|$ vertex disjoint paths  in $G[L^{*}]$ starting from every vertex in $L^{*}\cap R^{*}$ and ending at $L^{*}\cap T$.

\end{definition}
In this subsection, we will consider fixed $X$ and $T$ satisfying $X\subseteq T$. Therefore, we will use lean witnesses as an abbreviation of lean $(X,T,k')$-witnesses with arbitrary $k'$.

In \Cref{def:Update}, we formally define the carve operation that carves the terminal set $T$ with a lean witness. Next, \Cref{lemma:UpdateOneLeanWitness} shows that, roughly speaking, when we carve $T$ with a lean witness, the adhesion of $T$ will not increase and the size of $T$ will decrease. 

\begin{definition}[Carve with One Lean Witness]
Given a terminal set $T$, we define the \emph{carve} operation with respect to a lean witness $(L^*,R^*)$ to return the set $T' = (T \setminus L^*) \cup (L^* \cap R^*)$.
\label{def:Update}
\end{definition}

\begin{lemma}
\label{lemma:UpdateOneLeanWitness}
Suppose we carve $T$ with a lean-witness $(L^*,R^*)$. Then the adhesion of the new terminal set $T' = (T \setminus L^*) \cup (L^* \cap R^*)$ is at most $\max(|L^* \cap R^*|, \sigma(T))$ and $|T| - |T'| \geq \frac{|L^*\cap T|}{2\max\{1,|L^{*}\cap R^{*}|\}}$.%
\end{lemma}
\begin{proof}

Consider a component $C$ which is obtained after deleting $T'$ from $G$, and let us bound $|N(C)|$. There are two cases. In the first case, the component $C$ is a \emph{new component}, i.e. $C \subseteq L^* \setminus R^*$. In this case, clearly  $|N(C)|\le |L^* \cap R^*|$. In the second case, $C \subseteq R^* \setminus L^*$. Observe that in this case, $C \subseteq D$, where $D$ was a connected component of $G \setminus T$. Partition the boundary $N(C)$ into $N(C) \setminus N(D)$ and $N(C) \cap N(D)$.

First we claim that $N(C) \setminus N(D) \subseteq (L^* \cap R^*) \setminus T$, because $N(C)\setminus N(D) \subseteq L^{*}\cap R^{*}$ by the carving rule, and any $T$-vertex in $N(C)$ must belong to $N(D)$. Next, by \Cref{def:LeanWitness}, in $G[L^{*}]$, there exists $|L^{*}\cap R^{*}|$ vertex disjoint paths from $L^{*}\cap R^{*}$ to $L^{*}\cap T$. This means in $G[L^{*}\setminus (L^{*}\cap R^{*}\cap T)]$, there exists vertex-disjoint paths starting from all vertices in $(L^* \cap R^*) \setminus T$ to some vertices in $(L^* \setminus R^*) \cap T$. Consider the vertex-disjoint paths starting from $N(C)\setminus N(D)$. We know each of them must go through $N(D) \setminus N(C)$ by the following reasons. 
\begin{itemize}
\item Such paths end at $T$-vertices, so they must go through $N(D)$.
\item They cannot go through $N(D)\cap N(C)$ because $C\subseteq R^{*}\setminus L^{*}$ and $N(C)\cap N(D)\subseteq R^{*}\cap T$ which is disjoint from $L^{*}\setminus (L^{*}\cap R^{*}\cap T)$.
\end{itemize}
Therefore, it means that $|N(C) \setminus N(D)| \leq |N(D) \setminus N(C)|$. This in  turn yields
\begin{align*}
|N(C)| &= |N(C) \setminus N(D)| + |N(C) \cap N(D)| \\
&\leq |N(D) \setminus N(C)| + |N(C) \cap N(D)|\\
&= |N(D)| \leq \sigma(T).
\end{align*}

Lastly, observe that the reduction in the size of the terminal set $|T| - |T'|$ at least $|L^* \cap T| - |L^* \cap R^*|$. Let $a = |L^*\cap T|$ and $b = |L^* \cap R^*|$. If $a \geq 2b$, then $a - b \geq \frac{a}{2}$. Otherwise, because $(L^{\star},R^{\star})$ is a witness, $a - b \geq 1 \geq \frac{a}{2b}$. Therefore, $|T|-|T'|\geq \frac{|L^* \cap T|}{2\max\{1,|L^* \cap R^*|\}}$.
This completes the proof.
\end{proof}

We need a slight extension of the carve operation, as instead of updating the terminal set with one lean-witness, we will carve it using a collection of disjoint lean witnesses at once (see \Cref{def:DisjointCuts} to recall the definition of disjoint cuts).

\begin{definition}[Carve with Disjoint Lean Witnesses]
Given a terminal set $T$ and a collection of disjoint lean-witnesses $(L_i,R_i)$, $i \in [z]$, the carve operation with respect to this collection is defined as replacing $T$ by $T \cup (\bigcup_{i \in [z]} L_i \cap R_i) \setminus (\bigcup_{i \in [z]} (L_i \setminus R_i) \cap T)$.
\label{def:CarveWithDisjointCuts}
\end{definition}

\begin{lemma}\label{lemma:disjointlean}
Suppose we carve the terminal set $T$ with a set of disjoint lean witnesses $\{(L_i,R_i)\}$, $i \in [z]$. Then the adhesion of the new terminal set $T'$ is at most $\max\{\max_{i \in [z]} |L_i^* \cap R_i^*|, \sigma(T)\}$ and the size of $T$ reduces by at least $\sum_{i \in [z]} \frac{|L^*_i \setminus R^*_i|}{2\max\{1,|L^*_i \cap R^*_i|\}}$

\end{lemma}

\subsection{Covering Carvable Vertices with Disjoint Witnesses}

\Cref{coro:LeanWitCover} is the key lemma we prove in this subsection. Broadly, it shows that one can (a) efficiently find the set of all carvable vertices, and (b) find a set of disjoint lean witnesses that \emph{cover} a large fraction of carvable vertices (a vertex $v$ is covered by cut $(L,R)$ if $v\in L\setminus R$).

\begin{lemma}
\label{coro:LeanWitCover}
Given an $m$-vertex graph $G$, a terminal set $T$, a set $X\subseteq T$ and a parameter $k'$ satisfying $\sigma(T) \leq k'$, there is an algorithm running in time $2^{\OO(k' \log k')} m^{1+o(1)}$ that computes
\begin{itemize}
\item a set $Q\subseteq T\setminus X$ that includes all $(X,T,k')$-carvable vertices in $T$, and
\item a collection of disjoint lean $(X,T,k')$-witnesses ${\cal C}$ s.t. $(L\setminus R)\cap T\subseteq Q$ for each $(L,R)\in{\cal C}$, and
\[
\sum_{(L,R)\in{\cal C}} |(L\setminus R)\cap T|\geq |Q|/\gamma,
\]
where $\gamma = 2^{\OO(k'\log k')}\log n$.
\end{itemize}
\end{lemma}

In fact, \Cref{coro:LeanWitCover} is a simple corollary of the below \Cref{lemma:colorcoding} and \Cref{lemma:EnforceLean}, where the former will compute a collection disjoint witnesses (may not be lean) to cover a large fraction of carvable vertices, and the latter will convert disjoint witnesses to disjoint lean witnesses.

\begin{lemma}\label{lemma:colorcoding}
Given an $m$-vertex graph $G$, a terminal set $T$, a set $X\subseteq T$ and a parameter $k'$ satisfying $\sigma(T) \leq k'$, there is an algorithm running in time $2^{\OO(k' \log k')} m^{1+o(1)}$ that computes
\begin{itemize}
\item a set $Q\subseteq T\setminus X$ that includes all $(X,T,k')$-carvable vertices in $T$, and
\item a collection of disjoint $(X,T,k')$-witnesses ${\cal C}$ s.t. $(L\setminus R)\cap T\subseteq Q$ for each $(L,R)\in{\cal C}$, and
\[
\sum_{(L,R)\in{\cal C}} |(L\setminus R)\cap T|\geq |Q|/\gamma_{0},
\]
where $\gamma_{0} = 2^{\OO(k'\log k')}\log n$.
\end{itemize}

\end{lemma}
Most of this subsection is devoted to prove \Cref{lemma:colorcoding}. Before we start the proof, we explain on a high level how the algorithm works - for ease of understanding, we focus only on obtaining the set $Q$. Let $v \in T$ be an $(X,T,k')$-carvable vertex. Then, we know that there is a connected $(X,T,k')$ witness $(L,R)$ with $v \in L \setminus R$.

For simplicity, assume that $|(L \setminus R) \cap T| = k' + 1$ and $|(L \cap R) \cap T| = k'$. In~\Cref{lemma:bound} we show that there is a set $Y$ of at most $\OO(k'^3)$ vertices, such that for every edge $(u,t) \in E(H)$ where $u \in (L \setminus R)$ and $t \in (R \setminus L)$, we have $t \in Y$. Here $H = H_{T}$ is the torso of $T$ in $G$. Roughly speaking, if we use color coding such that every terminal in $(L \setminus R) \cap T$ is colored red, and every terminal in $Y$ is colored blue, then when we contract all the red vertices together, we can detect the cut $(L,R)$ by computing a min-cut between the contracted red vertices and the set $X$, which can be done using our single source min-cut subroutine,~\Cref{thm:SSMC}. The actual algorithm is slightly more complicated to account for the general case when $|(L \setminus R) \cap T| < k' + 1$ - in this case, we need to account for the terminals in $L \cap R \cap T$ as well. In this case we use color coding with $3$ colors, and modify the graph slightly before applying our result on single source min-cuts.

\begin{proof}
The proof combines the color coding technique in~\cite{cygan2020randomized} along with our single source mincuts subroutine, \Cref{thm:SSMC}. We remark that Theorem 6.1 of~\cite{korhonen2024minor} gives a similar result for a different notion of carvable vertices: they give a deterministic algorithm, but our randomized algorithm based on~\Cref{thm:SSMC} has a better dependence on $k'$. In this proof, we use $H$ to denote the torso $H_{T}$ of $T$ in $G$.

\paragraph{Fix the sets of terminals to be colored.} For the sake of analysis, for a vertex $v \in T$ which is $(X,T,k')$-carvable, we fix a \emph{connected} $(X,T,k')$-witness $(L^*_v, R^*_v)$ with $v\in L^{*}_{v}\setminus R^{*}_{v}$ (recall the definition of connected witness in \Cref{def:Carvable}) satisfying $N_{H}((L^{*}_{v}\setminus R^{*}_{v})\cap T) \supseteq L^{*}_{v}\cap R^{*}_{v}\cap T$. Such a witness must exist by \Cref{ob:BetterWitness} (we defer its proof to \Cref{sect:OmittedProof}). 

\begin{observation}
\label{ob:BetterWitness}
If there exists a connected $(X,T,k')$ witness $(L,R)$, there exists an connected $(X,T,k')$ witness $(L',R')$ with $(L'\setminus R')\cap T = (L\setminus R)\cap T$ further satisfying that $N_{H}((L'\setminus R')\cap T) \supseteq L'\cap R'\cap T$.
\end{observation}

Because $(L^{*}_{v},R^{*}_{v})$ is a connected $(X,T,k')$-witness, we know that $H[(L^*_v \setminus R^*_v) \cap T]$ is connected. We fix three vertex sets $Z_{v},Y_{v},W_{v}$ w.r.t. $v$ as follows.
\begin{itemize}
\item Let $Z_{v}$ be an arbitrary subset of $(L^{*}_{v}\setminus R^{*}_{v})\cap T$ of size $\min\{k'+1, |(L^{*}_{v}\setminus R^{*}_{v})\cap T|\}$ s.t. $H[Z_{v}]$ is connected (i.e. when $|(L^{*}_{v}\setminus R^{*}_{v})\cap T|\leq k'+1$, we take the entire $(L^{*}_{v}\setminus R^{*}_{v})\cap T$ as $Z_{v}$). 
\item Let $W_v = L^{*}_{v} \cap R^{*}_{v} \cap T$ be the terminals inside $L^{*}_{v}\cap R^{*}_{v}$.
\item Let $Y_v = N_{H}((L^{*}_{v}\setminus R^{*}_{v})\cap T)\setminus W_{v}$ be the $H$-neighbors of $(L^{*}_{v}\setminus R^{*}_{v})\cap T$ not falling in $L^{*}_{v}\cap R^{*}_{v}$.
\end{itemize}
Note that by definition, $|Z_v| \leq k' + 1$, $|W_v| \leq k'$, and by \Cref{lemma:bound}, $|Y_v| \leq |Z_v|(k'^2) \leq (k' + 1)k'^2 \leq 2k'^3$.

\begin{lemma}\label{lemma:bound}
Let $(L^*, R^*)$ be an $(X,T,k')$-witness. For every vertex $u \in (L^* \setminus R^*)\cap T$ there are at most $k'^2$ vertices $t \in (R^* \setminus L^*) \cap T$ such that $(u,t) \in E(H)$. 
\end{lemma}

\begin{proof}
Consider a vertex $t\in (R^{*}\setminus L^{*})\cap T$ s.t. $(u,t)\in E(H)$.
By the definition of $H$, it must be the case that there exists a path in $G$ from $u$ to $t$ with no internal vertex from $T$. Also since $u \in L^* \setminus R^*$ and $t \in R^* \setminus L^*$, this path must go through some vertex $x \in (L^* \cap R^*) \setminus T$. 
Let $C_x$ be the connected component of $G\setminus T$ containing $x$. Observe that $t \in N(C_x)$. 

The number of possible $t$ is at most $k'^2$ because there are $|(L^* \cap R^*) \setminus T| \le k'$ possible choices of $x$ and, for each $x$, $|N(C_x)| \le \sigma(T) \le k'$ by the assumption in \Cref{lemma:colorcoding}. 
\end{proof}

\paragraph{Color coding.} We are now ready to apply the color coding step. The following result, also used in~\cite{cygan2020randomized}, is essentially a deterministic fast algorithm for color coding.

\begin{lemma}[Lemma 2.2 of ~\cite{cygan2020randomized}, extension of~\cite{chitnis2016designing}]\label{lemma:hitmiss}
Given a set $U$ of size $n$ and integers $0 \leq a_1 \leq  a_2 \leq a _3 \leq a_{\ell} \leq n $, one can compute a family $\mathcal{F}$ of functions $f: U \rightarrow [\ell]$ of size $2^{\OO((a_1 + a_2 + \ldots a_{l-1}) \log (a_1 + a_2 + \ldots a_{\ell}))} \OO(\log^2 n)$ in time $2^{\OO((a_1 + a_2 + \ldots a_{\ell-1}) \log (a_1 + a_2 + \ldots a_{\ell}))} \OO(\log^2 n)$ such that for any pairwise disjoint sets $A_1, A_2, A_3 \ldots A_{\ell}$ of size at most $a_1, a_2, a_3 \ldots a_{\ell}$ respectively, there exists a function $f \in \mathcal{F}$ such that $f(x_i) = i$ for every $x_i \in A_i$, $i \in [\ell]$.   
\label{lemma:ColorFunctions}
\end{lemma}

We apply \Cref{lemma:ColorFunctions} on the universe $U = T$ with $a_{1} = k'+1$, $a_{2} = k'$ and $a_{3} = 2(k')^{3}$ to obtain a function family ${\cal F}$. For a function $f\in{\cal F}$, we use $f^{-1}(1), f^{-1}(2), f^{-1}(3)\subseteq T$ to denote the terminals with color $1$, $2$ and $3$ respectively. Furthermore, for a function $f\in {\cal F}$ and a terminal $v\in T$, we say that $v$ is \emph{lucky} w.r.t $f$ if $Z_{v}\subseteq f^{-1}(1)$,  $W_{v}\subseteq f^{-1}(2)$, and $Y_{v}\subseteq f^{-1}(3)$. Note that by the guarantee of \Cref{lemma:ColorFunctions}, for every $(X,T,k')$-carvable terminal $v\in T$, there must exist a function $f\in {\cal F}$ such that $v$ is lucky with respect to $f$. 

\paragraph{Compute single source mincuts.} Now for every $f \in \mathcal{F}$, we do the following.
We first construct a capacitated graph $\wtilde{G}$ with vertex capacity function $\rho$. 
\begin{itemize}
\item Start with the original graph $G$. For each non-terminal vertex $v\in V(G)\setminus T$, we set its capacity $\rho(v) = 1$. For each vertex $v\in T$, if $f(v) = 1$, set its capacity $\rho(v) 
= \infty$, otherwise (i.e. $f(v)=2$ or $f(v)=3$) set its capacity $\rho(v) = 1$.
\item For each connected component $C$ of $H[f^{-1}(1)]$, we add a super vertex $t_{C}$ with capacity $\rho(t_{C}) = \infty$ and add edges from $t_{C}$ to each vertex in $C\cup (N_{H}(C)\cap f^{-1}(2))$.
\item Add a super vertex $s$ with capacity $\rho(s) = \infty$, and add edges from $s$ to each vertex in $X$.
\end{itemize}
Note that a vertex cut $(\wtilde{L},\wtilde{R})$ in $\wtilde{G}$ (with finite cut size) naturally corresponds to a vertex cut $(L,R)$ in $G$ where $L$ and $R$ are obtained by dropping super vertices from $\wtilde{L}$ and $\wtilde{R}$ respectively. Therefore, when we have a cut $(\wtilde{L},\wtilde{R})$ (possibly with some subscripts and superscripts) in $\wtilde{G}$, we use $(L,R)$ (with the same subscripts and superscripts) to denote its corresponding cut in $G$.

On the graph $\wtilde{G}$, we run the single source mincut algorithm, \Cref{thm:SSMC}, with the vertex $s$ as the source, the vertices $T_{C} = \{t_{C}\mid \text{connected components $C$ of $H[f^{-1}(1)]$}\}$ as sinks, and the same parameter $k'$.
Let $T^{*}_{C} = \{t_{C}\in T_{C}\mid \lambda_{\wtilde{G}}(t_{C},s) \leq k'\}$.
The output is a mincut cover $\wtilde{\cal K}_{f}$ with respect to $s$ and $T^{*}_{C}$ in $\wtilde{G}$, and the width of $\wtilde{\cal K}_{f}$ is $O(k'\log^{3}n)$. The subscript $f$ of $\wtilde{\cal K}_{f}$ means $\wtilde{\cal K}_{f}$ is with respect to function $f$.

\paragraph{The final output for \Cref{lemma:colorcoding}.} Let ${\cal K}$ be obtained by replacing each cut $(\wtilde{L},\wtilde{R})\in \bigcup_{f\in{\cal F}} \wtilde{\cal K}_{f}$ with its corresponding cut $(L,R)$ in $G$. Then obtain another mincut cover ${\cal K}_{\wit}$ by keeping only those cuts of ${\cal K}$ which are  $(X,T,k')$-witnesses. The set $Q\subseteq T$ is defined as 
\[
Q = \bigcup_{(L,R)\in {\cal K}_{\wit}}(L\setminus R)\cap T.
\]
The collection ${\cal C}$ that will be output is the ${\cal C}\in{\cal K}_{\wit}$ that maximizes $|\bigcup_{(L,R)\in {\cal C}}(L\setminus R)\cap T|$.

\paragraph{Correctness of \Cref{lemma:colorcoding}.} First we show that, for each color function $f\in{\cal F}$ and each $(X,T,k')$-carvable vertex $v\in T$ that is lucky with respect to $f$, there is an $(X,T,k')$-witness $(L,R)\in{\cal K}_{f}$ with $v\in L\setminus R$. 

\begin{lemma}
\label{lemma:SSMCIsWitness}
Let $v \in T$ be a $(X,T,k')$-carvable vertex, and let $f\in{\cal F}$ be a function s.t. $v$ is lucky w.r.t. $f$. Let $C_{v}$ be the connected component of $H[f^{-1}(1)]$ containing $v$. We have
\begin{enumerate}
\item the size of $t_{C_{v}}$-$s$ mincut in $\wtilde{G}$, i.e. $\lambda_{\wtilde{G}}(t_{C_{v}}, s)$, is at most $k'$, and
\item any $t_{C_{v}}$-$s$ mincut $(\wtilde{L},\wtilde{R})$ in $\wtilde{G}$ corresponds to a $(X,T,k')$-witness $(L,R)$ in $G$ with $v\in L\setminus R$.
\end{enumerate}
\end{lemma}

We exploit the above \Cref{lemma:SSMCIsWitness} to prove this (we will prove \Cref{lemma:SSMCIsWitness} soon). First, we have $t_{C_{v}}\in T^{*}_{C}$ %
because $\lambda_{\wtilde{G}}(t_{C_{v}},s)\leq k'$ as statement 1 in \Cref{lemma:SSMCIsWitness} says. Because $\wtilde{\cal K}_{f}$ is a mincut cover with respect to $s$ and $T^{*}_{C}$ in $\wtilde{G}$, there is a $t_{C_{v}}$-$s$ mincut $(\wtilde{L},\wtilde{R})$ in $\wtilde{\cal K}_{f}$. By statement 2 in \Cref{lemma:SSMCIsWitness}, the cut in $G$ corresponding to $(\wtilde{L},\wtilde{R})$ is a $(X,T,k')$-witness with $v\in L\setminus R$.

Next, because every $(X,T,k')$-carvable vertex $v\in T$ is lucky w.r.t. some $f\in{\cal F}$ by \Cref{lemma:ColorFunctions}, we know $Q$ includes all $(X,T,k')$-carvable vertices in $T$. Also $Q$ is disjoint from $X$ because $L\setminus R$ is disjoint from $X$ for all $(L,R)\in{\cal K}$. Trivially ${\cal C}$ is a collection of disjoint cuts because each $\wtilde{\cal K}_{f}$ is a mincut cover. Lastly, we have $|\bigcup_{(L,R)\in{\cal C}}(L\setminus R)\cap T|\geq |Q|/O(|{\cal F}|\cdot k'\log^{3}n) \geq |Q|/(2^{\OO(k'\log k')\log n})$, because the width of ${\cal K}$ is $O(k'\log^{3}n)\cdot |{\cal F}|$.

\begin{proof}[Proof of \Cref{lemma:SSMCIsWitness}]
We prove statements 1 and 2 separately.

\medskip

\noindent{\underline{Statement 1.}} When $v$ is lucky with respect to $f$, by definition, the connected $(X,T,k')$-witness $(L^{*}_{v},R^{*}_{v})$ with $v\in L^{*}_{v}\setminus R^{*}_{v}$ we fixed above has $L^{*}_{v}\cap R^{*}_{v}\cap T = W_{v}\subseteq f^{-1}(2)$. By the construction of $\wtilde{G}$, vertices in $L^{*}_{v}\cap R^{*}_{v}$ are all have capacity $1$ in $\wtilde{G}$, so $\rho(L^{*}_{v}\cap R^{*}_{v}) = |L^{*}_{v}\cap R^{*}_{v}|\leq k'$.

Next, observe that removing $L^{*}_{v}\cap R^{*}_{v}$ will disconnect $t_{C_{v}}$ from $s$ in $\wtilde{G}$ due to the following reasons.
Recall that we only connect $t_{C_{v}}$ to $C_{v}\cup (N_{H}(C_{v})\cap f^{-1}(2))$ and only connect $s$ to $X$, so it suffices to show $C_{v}\cup (N_{H}(C_{v})\cap f^{-1}(2))\subseteq L^{*}_{v}$ and $X\subseteq R^{*}_{v}$. 
\begin{itemize}
\item We first prove $C_{v}\cup (N_{H}(C_{v})\cap f^{-1}(2))\subseteq L^{*}_{v}$. Note that $C_{v}\subseteq (L^{*}_{v}\setminus R^{*}_{v})\cap T$ because we have blocked all $H$-neighbors of $(L^{*}_{v}\setminus R^{*}_{v})\cap T$ using colors 2 and 3 (more precisely, we have $N_{H}((L^{*}_{v}\setminus R^{*}_{v})\cap T)\subseteq W_{v}\cup Y_{v}\subseteq f^{-1}(2)\cup f^{-1}(3)$). Furthermore, we have $N_{H}(C_{v})\cap f^{-1}(2)\subseteq L^{*}_{v}$ because $N_{H}(C_{v})\cap (R^{*}_{v}\setminus L^{*}_{v})\subseteq N_{H}((L^{*}_{v}\setminus R^{*}_{v})\cap T)\cap (R^{*}_{v}\setminus L^{*}_{v}) = N_{H}((L^{*}_{v}\setminus R^{*}_{v})\cap T)\setminus W_{v} = Y_{v}\subseteq f^{-1}(3)$.

\item We have $X\subseteq R^{*}_{v}$ because $(L^{*}_{v},R^{*}_{v})$ is an $(X,T,k')$-carvable witness.
\end{itemize}
Therefore, we can conclude that $\lambda_{\wtilde{G}}(t_{C_{v}}, s)\leq \rho(L^{*}_{v}\cap R^{*}_{v}) = |L^{*}_{v}\cap R^{*}_{v}|\leq k'$.

\medskip

\noindent{\underline{Statement 2.}} We will show that $(L,R)$ is an $(X,T,k')$-witness by verifying the properties stated in \Cref{def:Carvable} one by one (recall that $(L,R)$ is the cut in $G$ corresponding to $(\wtilde{L},\wtilde{R})$).

Properties \ref{prop:wit1} and \ref{prop:wit3} are easy to see.
\begin{itemize}
\item[1.] $|L\cap R|\leq k'$ is because $\lambda_{\wtilde{G}}(t_{C_{v}},s)\leq k'$
\item[3.] $X\subseteq R$ because when constructing $\wtilde{G}$, $s$ is connected to all vertices in $X$.
\end{itemize}
To see property \ref{prop:wit2}, i.e. $|L\cap T|>|L\cap R|$, we consider two cases. 

\medskip

{\underline{Case 1.}} The first case is $|Z_{v}| = k'+1$. Note that $Z_{v}\subseteq C_{v}$ (since $Z_{v}\subseteq f^{-1}(1)$ and $H[Z_{v}]$ is connected) and $C_{v}\subseteq (L\setminus R)\cap T$ (since all vertices in $f^{-1}(1)\supseteq C_{v}$ has infinite capacity in $\wtilde{G}$), so trivially we have $|(L\setminus R)\cap T| \geq |C_{v}|\geq |Z_{v}| = k'+1>k'\geq |(L\cap R)\setminus T|$, which is equivalent to $|L\cap T|>|L\cap R|$. 

\medskip

{\underline{Case 2.}} The second case is $Z_{v} = (L^{*}_{v}\setminus R^{*}_{v})\cap T$, which implies $C_{v} = (L^{*}_{v}\setminus R^{*}_{v})\cap T$ by $C_{v}\subseteq (L^{*}_{v}\setminus R^{*}_{v})\cap T$ (shown in the proof of statement 1) and $Z_{v}\subseteq C_{v}$. Recall that the witness $(L^{*}_{v},R^{*}_{v})$ we fixed above satisfies $N_{H}((L^{*}_{v}\setminus R^{*}_{v})\cap T) \supseteq L^{*}_{v}\cap R^{*}_{v}\cap T$, i.e., $W_{v}\subseteq N_{H}(C_{v})$.%
Combining $W_{v}\subseteq f^{-1}(2)$, we have $W_{v}\subseteq N_{H}(C_{v})\cap f^{-1}(2)$. Thus, by the construction of $\wtilde{G}$, we have $W_{v}\subseteq L$ (since $t_{C_{v}}$ is connected to $W_{v}$) and $Z_{v}=C_{v}\subseteq L\setminus R$ (since $t_{C_{v}}$ is connected to $Z_{v}$, and $Z_{v}$-vertices have infinite capacity), which implies
\[
|L\cap T|\geq |W_{v}| + |Z_{v}| = |L^{*}_{v}\cap T|.
\]
Therefore, we have
\[
|L\cap T|\geq |L^{*}_{v}\cap T|>|L^{*}_{v}\cap R^{*}_{v}|\geq \lambda_{\wtilde{G}}(t_{C_{v}},s) = |L\cap R|,
\]
where $|L^{\star}_{v}\cap T|>|L^{\star}_{v}\cap R^{\star}_{v}|$ is because $(L^{\star}_{v}, R^{\star}_{v})$ is an $(X,T,k')$-witness.

Finally, we have $v\in C_{v}\subseteq L\setminus R$.
\end{proof}

\paragraph{Running time of \Cref{lemma:colorcoding}.}
The bottleneck is to call the single-source mincut subroutine \Cref{thm:SSMC} for each function $f\in{\cal F}$. The total running time is $|{\cal F}|\cdot k' m^{1+o(1)} = 2^{\OO(k'\log k')}m^{1+o(1)}$.
This completes the proof of \Cref{lemma:colorcoding}.
\end{proof}

Finally, we obtain \Cref{coro:LeanWitCover} from  \Cref{lemma:colorcoding} by making each witness from \Cref{lemma:colorcoding} to be lean using the following lemma.

\begin{lemma}
Given an $m$-vertex graph $G$, a set of terminals $T$ and a collection of disjoint $(X,T,k')$-witnesses ${\cal C}$, there is an algorithm that computes, for each $(L,R)\in{\cal C}$, a lean $(X,T,k')$-witness $(L',R')$ satisfying that $L'\setminus R'\subseteq L\setminus R$ and $|(L'\setminus R')\cap T|\geq |(L\setminus R)\cap T|/(k'+1)$. Furthermore, these new lean witnesses are disjoint. The running time is $O(k'm)$.
\label{lemma:EnforceLean}
\end{lemma}
\begin{proof}
For each cut $(L,R)\in{\cal C}$, we do the following. First, we construct a unit-capacitated graph $\wtilde{G}_{\local}$. Initially, $\wtilde{G}_{\local} = G[L]\setminus E(G[L\cap R])$ and assign each vertex with capacity $1$. Next, we create a super vertex $v_{\src}$ as source and connect $v_{\src}$ to each vertex in $L\cap T$. Also, we create a super vertex $v_{\sink}$ as sink and connect $v_{\sink}$ to each vertex in $L\cap R$. 

We compute an arbitrary $v^{\src}$-$v^{\sink}$ mincut $(\wtilde{L},\wtilde{R})$ in $\wtilde{G}_{\local}$. To avoid clutter, in what follows we suppose $v^{\src},v^{\sink}$ have been removed from $\wtilde{L}$ and $\wtilde{R}$. We define $(L',R')$ by letting $L' = \wtilde{L}$ and $R' = \wtilde{R}\cup (R\setminus L)$. Note that by the construction of $\wtilde{G}_{\local}$, the mincut size $|\wtilde{L}\cap \wtilde{R}|\leq |L\cap R|$. We now show $(L',R')$ satisfies the requirements. We start with some simple observation.
\begin{enumerate}
\item\label{ob:MakeLean1} We have $L\cap T\subseteq \wtilde{L}$, because $v_{\src}$ connects to $L\cap T$.
\item\label{ob:MakeLean2} We have $L\cap R\subseteq \wtilde{R}$, because $v^{\sink}$ connects to $L\cap R$.
\item\label{ob:MakeLean3} We have $|\wtilde{L}\cap \wtilde{R}|\leq |L\cap R|$ by the construction of $\wtilde{G}_{\local}$.
\end{enumerate}
Note that $(L',R')$ is indeed a cut in $G$ because $L'\setminus R'$ has no edge to $R'\setminus L'$. Concretely, $R'\setminus L' = (\wtilde{R}\setminus \wtilde{L})\cup (R\setminus L)$, and $L'\setminus R'$ has no edge to $\wtilde{R}\setminus \wtilde{L}$ (because $L'\setminus R' = \wtilde{L}\setminus \wtilde{R}$) and $R\setminus L$ (because $\wtilde{L}\setminus \wtilde{R} \subseteq L\setminus R$ by observation \ref{ob:MakeLean2}).

We first show $L'\setminus R'\subseteq L\setminus R$. Again we have $L'\setminus R' = \wtilde{L}\setminus \wtilde{R}$ by the definition of $(L',R')$. Combining $\wtilde{L}\subseteq L$ and $L\cap R\subseteq \wtilde{R}$, we have $\wtilde{L}\setminus \wtilde{R}\subseteq L\setminus R$.

Second, we show that $(L',R')$ is a $(X,T,k')$-witness. First, we have $|L'\cap R'| = |\wtilde{L}\cap \wtilde{R}|\leq |L\cap R|\leq k'$ by observation \ref{ob:MakeLean3}. Next, $|L'\cap T| = |\wtilde{L}\cap T|\geq |L\cap T|>|L\cap R|\geq |L'\cap R'|$, because $(L,R)$ is an $(X,T,k')$-witness and observations \ref{ob:MakeLean1} and \ref{ob:MakeLean3}. Finally, we have $X\subseteq R'$ because $L'\setminus R'\subseteq L\setminus R$.

Third, $(L',R')$ is lean by the following reasons. Consider the maxflow in $\wtilde{G}_{\local}$ that certifies $(\wtilde{L},\wtilde{R})$ is a $v^{\src}$-$v^{\sink}$ mincut. The flow paths correspond to $|\wtilde{L}\cap \wtilde{R}|$ vertex disjoint paths in $G[\wtilde{L}]$ from $L\cap T$ (recall that $v^{\src}$ connects to $L\cap T$) to $\wtilde{L}\cap \wtilde{R}$. %
By the definition of $(L',R')$, these vertex disjoint paths show that $(L',R')$ is lean.

Fourth, we show $|(L'\setminus R')\cap T|\geq |(L\setminus R)\cap T|/(k'+1)$. Because $\wtilde{L}\setminus \wtilde{R} = L\setminus R$, it suffices to show that $|(\wtilde{L}\setminus \wtilde{R})\cap T|\geq |(L\setminus R)\cap T|/(k'+1)$. By observation \ref{ob:MakeLean1}, we have $L\cap T = \wtilde{L}\cap T$. Note that $|(\wtilde{L}\setminus \wtilde{R})\cap T| \geq |\wtilde{L}\cap T| - |\wtilde{L}\cap \wtilde{R}|$. Let $a = |\wtilde{L}\cap T| = |L\cap T|$ and $b = |\wtilde{L}\cap \wtilde{R}|$. We have $a = |L\cap T| > |L\cap R| \geq b$ and $b \leq |L\cap R|\leq k'$. Therefore, from $a\geq b+1$ and $b\leq k'$, we have $a-b\geq a/(k'+1)$. Namely, $|(\wtilde{L}\setminus \wtilde{R})\cap T|\geq |L\cap T|/(k'+1) \geq |(L\setminus R)\cap T|/(k'+1)$

Finally, these new lean witnesses are disjoint because for each $(L,R)\in{\cal C}$, its new lean witness
$(L',R')$ has $L' = \wtilde{L}\subseteq L$ and $L'\setminus R' = \wtilde{L}\setminus \wtilde{R}\subseteq L\setminus R$. It is easy to check that the new lean witnesses satisfy \Cref{def:DisjointCuts}.

The whole algorithm just runs maxflow on the graphs $\wtilde{G}_{\local}$ for each cut $(L,R)\in{\cal C}$. Because ${\cal C}$ is a collection of disjoint cuts, the graphs $\wtilde{G}_{\local}$ are edge-disjoint, so we only runs maxflow on graphs of total size $O(m)$. Since all the maxflow value is bounded by $k'$ by observation \ref{ob:MakeLean3}, the can use the classic Ford-Fulkerson maxflow algorithm. The total running time is $O(k'm)$.
\end{proof}

\subsection{The Algorithm for Reducing Adhesion: Proof of \Cref{thm:reducing_adhesion}}
\label{sec:algo_reducing_adhesion}

The algorithm for \Cref{thm:reducing_adhesion} is described in \Cref{alg:reducing_adhesion}. At a high level, at any time, we maintain a set of terminals $T$ and a set $X$. Initialize $T = V(G)$ and $X = X_0$. During the algorithm, $X$ keeps growing, and $T$ is roughly shrinking. Additionally, we maintain the invariant that $T \supseteq X$, $X$ is appropriately unbreakable, and $T$ has small adhesion. Once $T = X$, we obtain an unbreakable $X$ with small adhesion as desired.

\begin{algorithm}
\caption{Reducing Adhesion}
\label{alg:reducing_adhesion}
\begin{algorithmic}[1]
\Require A graph $G$ with a paramter $\epsilon > 0$ and a $(q,k)$-unbreakable set $X_{0}\subseteq V(G)$.
\Ensure A $({k} \lceil \frac{1}{\epsilon}\rceil + q, k)$ unbreakable set $X$ s.t. $X_{0} \subseteq X\subseteq V(G)$  and $\sigma(X) \leq {k} \lceil \frac{1}{\epsilon}\rceil + q$.
\State Initialize $X \gets X_0$ and $T \gets V(G)$
\For {$\ell = 1$ to $\lceil \frac{1}{\epsilon} \rceil$}
\State\label{line:colorcoding1} Apply \Cref{coro:LeanWitCover} on $G,T,X$ and $k'=q + \ell k$. Let $Q$ and ${\cal C}$ be the output.
\While {$|Q| \geq n^{1 - \ell\epsilon}$}
\State $T\gets$ carve $T$ with ${\cal C}$ as in \Cref{def:CarveWithDisjointCuts}.
\State\label{line:colorcoding2}Recompute $Q$ and ${\cal C}$ by applying \Cref{coro:LeanWitCover} on $G$, the new $T$, $X$, and $k'=q+\ell k$.
\EndWhile
\State $X \gets T\setminus Q$
\EndFor
\State Return $X$
\end{algorithmic}
\end{algorithm}

To prove the correctess of \Cref{alg:reducing_adhesion}, we will show the invariant that, at the beginning of each phase $\ell$, 
\begin{enumerate}
\item $X_{0}\subseteq X\subseteq T$,
\item $X$ is $(q+ (\ell-1) k,k)$-unbreakable,
\item $T$ has adhesion $\sigma(T)\leq q + (\ell-1) k$, and
\item $|T \setminus X|\leq n^{1-(\ell-1)\epsilon}$.
\end{enumerate}
Note that at the end of the algorithm, we have $X = T$ by the stopping condition of the inner loop. Therefore, once we prove the invariant, we immediately show the correctness of the final output $X$.

Let us give an intuition about the properties of $X$. On every phase, we always have $X_0 \subseteq X \subseteq T$. When the algorithm proceeds to the next phase, the unbreakability of $X$ and adhesion of $T$ are slightly relaxed by an additive factor of $k$. But the set $T\setminus X$ significantly shrinks by a factor of $n^\epsilon$. Intuitively, $T\setminus X$ contains set of vertices that we might carve out from $T$. Once $T\setminus X = \emptyset$, we have $X = T$ and we are done.\footnote{This dynamics resembles how the expander decomposition algorithms in \cite{nanongkai2017dynamic,wulff2017fully} work. At the beginning of each phase in these algorithms, the expansion parameter is slightly relaxed but the maximum balance of sparse cuts shrinks significantly.}

In what follows, we prove the invariant by induction. Trivially, the invariant holds at the beginning of the first phase $\ell=1$. Now fix a phase $1\leq \ell\leq \lceil 1/\epsilon\rceil$, assuming the invariant holds at the beginning of this phase, we show that it holds at the beginning of next phase $\ell+1$.

Each time we compute $Q$ and ${\cal C}$ by applying \Cref{coro:LeanWitCover} (i.e. in \Cref{line:colorcoding1} and \Cref{line:colorcoding2}), we have the following two cases and we first discuss some useful observations in these two cases.

\medskip

\noindent\underline{Case 1: $|Q|\geq n^{1-\ell\epsilon}.$} In this case we carve $T$ with ${\cal C}$. By \Cref{coro:LeanWitCover}, ${\cal C}$ is a collection of disjoint lean $(X,T,q+\ell k)$-witnesses. Let $T'$ be the new $T$ by carving the old $T$ with ${\cal C}$ according to \Cref{def:CarveWithDisjointCuts}. Note that we have $\sum_{(L,R)\in{\cal C}}|(L\setminus R)\cap T|\geq |Q|/\gamma$ by \Cref{coro:LeanWitCover} where
\[
\gamma = 2^{\OO((q + \ell k)\log (q + \ell k))}\log n,
\]
and each $(L,R)\in{\cal C}$ has $|L\cap R|\leq q + \ell k$ because $(L,R)$ is a lean $(X,T, q+ \ell k)$-witness.
By \Cref{lemma:disjointlean}, we have $T'\supseteq X$, $\sigma(T') \leq \max\{\sigma(T),q+ \ell k\}$ and 
\[
|T'|-|T|\geq \sum_{(L,R)\in{\cal C}}\frac{|L\setminus R|}{2|L\cap R|}\geq \sum_{(L,R)\in{\cal C}}\frac{|(L\setminus R)\cap T|}{2|L\cap R|}\geq \frac{|Q|}{2\gamma(q+ \ell k)}\geq \frac{n^{1-\ell\epsilon}}{2\gamma(q + \ell k)}.
\]

\medskip

\noindent\underline{Case 2: $|Q|<n^{1-\ell\epsilon}$.} In this case, we grow $X$ to be $T\setminus Q$. We use $X$ to denote the old $X$ before update, and let $X'$ be the new $X$ after the update. By \Cref{coro:LeanWitCover}, $Q\subseteq T\setminus X$ and $Q$ includes all $(X,T,q + \ell k)$-carvable vertices in $T$, so there is no $(X,T,q + \ell k)$-carvable vertex in $X'$.

\medskip

Now we verify the invariant at the end of phase $\ell$ (i.e. at the beginning of phase $\ell+1$). Again, we let $X$ denote the old $X$ at the beginning of phase $\ell$, and let $X'$ denote the new $X$ at the end of this phase.
\begin{enumerate}
\item We have $X'\subseteq T$ because  $X'=T\setminus Q$. Also, $X \subseteq X'$  because $Q\subseteq T\setminus X$. Thus, by induction, $X_0 \subseteq X \subseteq X'$.
\item Because $X'$ has no $(X,T,q + \ell k)$-carvable vertex and $X$ is $(q + (\ell-1) k,k)$-unbreakable by assumption, $X'$ is $(q + \ell k,k)$-unbreakable by \Cref{lemma:nocarvable}. 
\item By assumption, at the beginning of this phase, the adhesion of $T$ is at most $q + (\ell - 1)k$. Each time we carve $T$ into a new one (denoted by $T'$), we have $\sigma(T') \leq \max\{\sigma(T),q + \ell k\}$ by \Cref{lemma:disjointlean}. Thus, the final $T$ at the end of this phase has adhesion at most $q + \ell k$.
\item At the end of this phase, $|T\setminus X'|$ is exactly the size of $Q$ in case 2, so $|T \setminus X'|\leq n^{1-\ell \epsilon}$.
\end{enumerate}

\begin{lemma}
~\Cref{alg:reducing_adhesion} runs in time $\exp(\OO((q+\frac{k}{\epsilon})\log(q+\frac{k}{\epsilon}))) m^{1+\epsilon+o(1)}$.
\end{lemma}

\begin{proof}
In each phase $\ell$, we claim that case 1 will occur at most $O(\gamma(q + \ell k)n^{\epsilon})$ times due to the following reasons. At the beginning of phase $\ell$, we have $|T|-|X|\leq n^{1-(\ell-1)\epsilon}$. We have shown that each time we carve $T$ in case 1, the size of $T$ will drop by at least $\frac{n^{1-\ell\epsilon}}{2\gamma(q+ \ell k)}$. Therefore, after at most $2\gamma(q+\ell k)n^{\epsilon}$ occurrences of case 1, the differences between the size of $T$ and the size of (the old) $X$ will drop below $n^{1-\ell\epsilon}$. Once this happens, the next iteration of the inner loop will go into case 2 because $Q\subseteq T\setminus X$.

Therefore, there are at most $2\gamma(q+\ell k)n^{\epsilon} = \exp(\OO((q+\frac{k}{\epsilon})\log(q+\frac{k}{\epsilon})))n^{\epsilon}$ iterations of the inner loop in one phase. Each iteration runs in time $\exp(\OO((q+\frac{k}{\epsilon})\log(q+\frac{k}{\epsilon}))) m^{1+o(1)}$ (since the bottleneck is \Cref{lemma:colorcoding}) and thus the total running time of one phase is at most $\exp(\OO((q+\frac{k}{\epsilon})\log(q+\frac{k}{\epsilon}))) m n^{\epsilon+o(1)}$. Because there are at most $O(1/\epsilon)$ phases, the total run-time of the procedure is $\exp(\OO((q+\frac{k}{\epsilon})\log(q+\frac{k}{\epsilon}))) m n^{\epsilon+o(1)}$ as well. 
\end{proof}

\section{The Unbreakable Decomposition Algorithm}\label{sec:Decomp}

In this section, we show the main result of our paper, a fast algorithm to construct unbreakable decomposition.

\begin{theorem}
\label{thm:DetailedMain}
Given an $n$-vertex $m$-edge undirected graph $G$ with a parameters $k$ and $0<\epsilon\leq 1$, there is a randomized algorithm that, with high probability, computes a $(2\lceil 1/\epsilon \rceil k + 3k,k)$-unbreakable decomposition with adhesion $2\lceil 1/\epsilon \rceil k + 2k$, which further has the following properties.  
\begin{itemize}
\item The unbreakable decomposition is compact and admits subtree unbreakability.
\item The depth of the decomposition is $O(k\log n/\epsilon)$.
\item The number of tree nodes is $O(n)$ and the total bag size is $O(kn/\epsilon)$.
\end{itemize}
The running time is $2^{O(\frac{k}{\epsilon}\log\frac{k}{\epsilon})}m^{1+\epsilon+o(1)}$. Alternatively, the algorithm can compute an $(O(k/\epsilon),k)$-unbreakable decomposition with adhesion $O(k/\epsilon)$, depth $O(\log n)$, and all the other properties as above in asymptotically the same running time.
\end{theorem}

We now show the algorithm for computing an unbreakable decomposition with $O(k\log(n)/\epsilon)$ depth. Later in \Cref{sect:DepthReduction}, we will discuss how to reduce the depth to $O(\log n)$ by modifying the algorithm slightly.

\begin{proof}[Proof of \Cref{thm:DetailedMain}]
We fix a parameter $\sigma = \lceil\frac{1}{\epsilon}\rceil k+k$ (which is an upper bound of the adhesion of the set returned by \Cref{thm:reducing_adhesion}). We will design a subroutine $\Decomp(H,B)$, which will receive as input an undirected graph $H$ with a boundary vertex set $B\subseteq V(H)$ such that
\begin{enumerate}
\item\label{cond:Decomp1} $|B| \leq 2\sigma$,
\item\label{cond:Decomp2} $H\setminus B$ is connected, and
\item\label{cond:Decomp3} $N_{H}(V(H)\setminus B) = B$.
\end{enumerate}
We remark that condition \ref{cond:Decomp1} is relatively important, while conditions 2 and 3 are only for avoiding some corner cases (so the reader may ignore them for now and come back when they are used). 
The output of $\Decomp(H,B)$ is an unbreakable decomposition $(T,\beta)$ of $H$ in which the root bag $\beta(r)$ contains $B$. Then the desired unbreakable decomposition of $G$ can be obtained by invoking $\Decomp(G,\emptyset)$. Without loss of generality, we assume $G$ is connected (otherwise, just work on each connected component separately). Thus it satisfies conditions \ref{cond:Decomp2} and \ref{cond:Decomp3}. 

In what follows, we first describe $\Decomp(H,B)$ and then show the output guarantees of this subroutine.

\begin{algorithm}[h]
\caption{$\Decomp(H,B)$}
\label{alg:Decomp}
\begin{algorithmic}[1]
\Require Parameters $k\geq 1, 0<\epsilon\leq 1$ and $\sigma = \lceil1/\epsilon\rceil k + k$. A graph $H$ with a boundary vertex set $B\subset V(H)$ such that $|B|\leq 2\sigma$.
\Ensure A rooted tree decomposition $(T,\beta)$ of $H$ such that $B$ is contained in the root bag.
\If{$\sigma + 1\leq |B|\leq 2\sigma$}
\State\label{line:CheckUnbreakable}Apply \Cref{lemma:CheckUnbreakable} on $B$ and $H$ to check if $B$ is $(k,k)$-unbreakable in $H$.
\If{\Cref{lemma:CheckUnbreakable} returns a $(k,k)$-breakable witness $(L,R)$ of $B$ in $H$}
\State $X_{1}\gets L\cap R$ and $X\gets B\cup X_{1}$.
\Else \Comment{$B$ is $(k,k)$-unbreakable in $H$}
\State $X\gets $ apply \Cref{thm:reducing_adhesion} on graph $H$ with $B$ as the initial set.
\EndIf
\Else \Comment{$|B|\leq \sigma$}
\State $X_{1}\gets$ apply \Cref{coro:UnbreakLowAdhesion} on graph $H$ with parameters $\epsilon,k$.
\State $X\gets B\cup X_{1}$.
\EndIf
\If{$X=B$}
\State\label{line:corner}Return a trivial decomposition $(T,\beta)$ with a single bag $V(H)$.
\Else \Comment{$X\supset B$.}
\State Initialize $(T,\beta)$ with only the root node $r$ whose bag is $\beta(r) = X$.
\For{each connected component $D$ of $H\setminus X$}
\State $(T_{D},\beta_{D})\gets \Decomp(H[D\cup N_{H}(D)]\setminus E(H[N_{H}(D)]), N_{H}(D))$
\State Update $(T,\beta)$ by appending $(T_{D},\beta_{D})$ to the root node $r$.
\EndFor
\State Return $(T,\beta)$.
\EndIf
\end{algorithmic}
\end{algorithm}

\paragraph{Description of $\Decomp(H,B)$.} See \Cref{alg:Decomp} for a pseudocode. We now explain the algorithm below. 

\medskip

\noindent{\underline{Compute a ``separator'' $X$.}} We first compute a set $X\subseteq V(G)$ according to the following two cases.

\medskip

\noindent{\textbf{Case 1.}} If $\sigma+1\leq |B|\leq 2\sigma$, we will compute an unbreakable set $X\supseteq B$ with adhesion smaller than $|B|$ to reduce the boundary size of the subproblems.
\begin{itemize}
\item[(a)] If $B$ has a $(k,k)$-breakable witness in $H$, find such a witness $(L,R)$ using \Cref{lemma:CheckUnbreakable} (let $X_{1} = L\cap R$), and we let $X = X_{1}\cup B$.
\item[(b)] Otherwise, $B$ is $(k,k)$-unbreakable in $H$. We apply \Cref{thm:reducing_adhesion} on graph $H$ with $B$ as the initial set, and let $X$ be the output set.
\end{itemize}

\begin{claim}
\label{claim:DecompCase1}
In Case 1, the set $X$ is $(2\lceil 1/\epsilon\rceil k+3k,k)$-unbreakable and has adhesion at most $|B|-1$.
\end{claim}
\begin{proof}
We first focus on case (a). $X$ is trivially $(2\lceil 1/\epsilon\rceil k+3k,k)$-unbreakable because $|X|\leq |X_{1}| + |B|\leq k + 2\sigma = 2\lceil1/\epsilon\rceil k + 3k$. To see that $X$ has adhesion $|B|-1$, consider an arbitrary connected component $D$ of $H\setminus X$. Recall that $X = X_{1}\cup B$, so
\[
|N_{H}(D)| = |N_{H}(D)\cap X_{1}| + |N_{H}(D)\cap (B\setminus X_{1})|,
\]  
and we will bound $|N_{H}(D)\cap X_{1}|$ and $|N_{H}(D)\cap (B\setminus X_{1})|$ separately. Observe that $|N_{H}(D)\cap X_{1}|\leq |X_{1}|\leq k$, where $|X_{1}|\leq k$ is trivially because $X_{1}$ is a $(k,k)$-breakable witness. To bound $|N_{H}(D)\cap (B\setminus X_{1})|$, we look at the connected component $D'\supseteq D$ of $H\setminus X_{1}$ (note that $D'$ is unique since $X_{1}\subseteq X$). We have $N_{H}(D)\cap (B\setminus X_{1})\subseteq D'\cap B$ because $N_{H}(D')\subseteq X_{1}$. Furthermore, $|D'\cap B|\leq |B|-(k+1)$ because $X_{1}$ is a $(k,k)$-breakable witness of $B$ in $H$. Therefore, $|N_{H}(D)\cap (B\setminus X_{1})|\leq |B|-(k+1)$. Putting it all together, $|N_{H}(D)| \leq k + |B|-(k+1)\leq |B|-1$ as desired.

For (b), the claim holds simply by the guarantees of \Cref{thm:reducing_adhesion}, which guarantees that $X$ is $(k\lceil\frac{1}{\epsilon}\rceil+k,k)$-unbreakable and has $\sigma(X) = k\lceil\frac{1}{\epsilon}\rceil+k$ (which means $\sigma(X)\leq \sigma\leq |B|-1$). 
\end{proof}

\medskip

\noindent{\textbf{Case 2.}} If $|B|\leq \sigma$, we will compute a balanced unbreakable set $X\supseteq B$ (with adhesion at most $2\sigma$) to significantly reduce the graph size of the subproblems. We first compute $X_{1}$ by applying \Cref{coro:UnbreakLowAdhesion} on graph $H$ with parameters $\epsilon,k$, and then let $X$ be $B\cup X_{1}$.

\begin{claim}
In Case 2, the set $X$ of $H$ is $(2\lceil1/\epsilon\rceil k+2k,k)$-unbreakable and $1/2$-balanced, and has adhesion at most $2\lceil 1/\epsilon\rceil k + 2k$.
\label{claim:DecompCase2}
\end{claim}
\begin{proof}

\Cref{coro:UnbreakLowAdhesion} guarantees that, with high probability, $X$ is $(\lceil 1/\epsilon\rceil k+k,k )$-unbreakable with adhesion at most $\lceil 1/\epsilon\rceil k+k$.
Because $X = X_{1}\cup B$, $X$ is an $(\lceil 1/\epsilon\rceil k+k + |B|,k)$-unbreakable and $1/2$-balanced set in $H$ with adhesion at most $\lceil 1/\epsilon\rceil k+k + |B|$. Recall that $|B|\leq \sigma = \lceil 1/\epsilon\rceil k + k$, so the claim holds.
\end{proof}

\noindent{\underline{Construct the decomposition.}} After computing the set $X$, we construct the decomposition $(T,\beta)$ as follows. We first consider the corner case that $X = B$, in which we return a trivial decomposition with $V(H)$ as the only bag. 

\begin{claim}
In the corner case $X = B$, we have $|V(H)|\leq 2\sigma$, which implies $V(H)$ is $(2\lceil 1/\epsilon\rceil k + 2k,k)$-unbreakable in $H$.
\label{claim:CornerCase}
\end{claim}
\begin{proof}
Note that only case 2 can lead to this corner case. Because from the input conditions \ref{cond:Decomp2} and \ref{cond:Decomp3}, the adhesion of $X$ is exactly $|B|$, contradicting \Cref{claim:DecompCase1} for case 1. Therefore, $|X| = |B|\leq \sigma$ and \Cref{claim:DecompCase2} tells that $X$ is $1/2$-balanced. Combining that there is only one connected component $D$ in $H\setminus X$ (by input condition \ref{cond:Decomp2}), we have $|V(H)|\leq |D| + |X|\leq |V(H)|/2 + \sigma$, which means $|V(H)| \leq 2\sigma$. Hence $V(H)$ is trivially $(2\sigma,k)$-unbreakable in $H$.
\end{proof}

From now we assume $X\supset B$. We consider each connected component $D$ of $H\setminus X$, and solve a subproblem $\Decomp(H[D\cup N_{H}(D)]\setminus E(H[N_{H}(D)]),N_{H}(D))$\footnote{In other words, $H[D\cup N_{H}(D)]\setminus E(H[N_{H}(D)])$ is the subgraph of $H$ induced by vertices $D\cup N_{H}(D)$ excluding edges with both endpoints in $N_{H}(D)$.} recursively which gives output $(T_{D},\beta_{D})$ (note that this will ensure the graphs of the subproblems are edge-disjoint). After solving all the subproblems, we create a node $r$ as the root of $T$ corresponding to bag $\beta(r) = X$, and then append all $(T_{D},\beta_{D})$ to $r$ as \emph{child-subtrees}.

\paragraph{Correctness.} Let $(T,\beta)$ be the output of $\Decomp(H,B)$.

\medskip

\noindent{\underline{$(T,\beta)$ is an unbreakable decomposition.}} We now prove by induction that $(T,\beta)$ from $\Decomp(H,B)$ is a $(2\lceil 1/\epsilon \rceil k + 3k,k)$-unbreakable decomposition of $H$ with adhesion $2\lceil 1/\epsilon \rceil k + 2k$, further satisfying that the root bag containing $B$.

We first consider the base case that $\Decomp(H,B)$ causes no further recursion, which means the set $X$ is trivially the whole $V(H)$. Then trivially the output $(T,\beta)$ is a tree decomposition with only a single bag, and this bag is $(O(k/\epsilon),k)$-unbreakable in $H$ by \Cref{claim:DecompCase1} and \Cref{claim:DecompCase2}.

Next, we consider the recursive subproblems that $\Decomp(H,B)$ invokes. Recall that each recursion $\Decomp(H[D\cup N_{H}(D)]\setminus E(H[N_{H}(D)]), N_{H}(D))$ relates to a connected component $D$ of $H\setminus X$, and we assume all these recursions output child-subtrees $(T_{D},\beta_{D})$ correctly. We verify the properties in \Cref{def:TreeDecomp,def:UnbreakableDecomp} one by one to show the correctness of $\Decomp(H,B)$.
\begin{itemize}
\item Consider a vertex $v\in V(H)$. If $v\notin X$, then $v$ is inside a unique connected component $D_{v}$ of $H\setminus X$, so only the child-subtree $(T_{D_{v}}, \beta_{D_{v}})$ may have bags containing $v$, and by the correctness of recursions, bags containing $v$ induce a connected subtree. If $v\in X$, any child-subtree $(T_{D},\beta_{D})$ with bags containing $v$ must relate to a component $D$ s.t. $v\in N_{H}(D)\subseteq X$. Hence, $v$ is inside the root bag of such child-subtrees $(T_{D},\beta_{D})$ (since $N_{H}(D)$ are the boundary vertices of the next recursion). By the construction of $(T,\beta)$, bags containing $v$ induce a connected subtree.

\item Each edge in $E(H)$ will be covered by at least a bag in $(T,\beta)$, because edges in $H[X]$ are inside the root bag of $(T,\beta)$, and each of the rest of edges is inside some $H[D\cup N_{H}(D)]\setminus E(H[N_{H}(D)])$.

\item All bags of $(T,\beta)$ are $(2\lceil 1/\epsilon \rceil k + 3k,k)$-unbreakable in $H$ by the following reasons, The new bag $X$ is $(2\lceil 1/\epsilon \rceil k + 3k,k)$-unbreakable in $H$ by \Cref{claim:DecompCase1} and \Cref{claim:DecompCase2}. each old bag is even $(2\lceil 1/\epsilon \rceil k + 3k,k)$-unbreakable in some subgraph of $H$.

\item To see that $(T,\beta)$ has adhesion $2\sigma = 2\lceil 1/\epsilon \rceil k + 2k$, it suffices to show that the root bag $X$ of $(T,\beta)$ has an intersection of size $2\sigma$ with the root bag of each child-subtree $(T_{D},\beta_{D})$. Indeed, observe that the intersection is exactly $N_{H}(D)$, so it has size at most $2\sigma$ (since $X$ has adhesion $2\sigma$ in $H$ by \Cref{claim:DecompCase1} and \Cref{claim:DecompCase2}).
\end{itemize}

\medskip

\noindent\underline{The tree size and the total bag size.} Next, we prove that $(T,\beta)$ from $\Decomp(H,B)$ satisfies (1) $|V(T)|\leq |V(H)|$, (2) the total bag size is at most $O(|V(H)|\sigma) = O(|V(H)|k/\epsilon)$. Recall that the definition of tree decomposition guarantees the $\beta(t)\setminus \sigma(t)$ for all $t\in V(T)$ are disjoint. Hence (1) holds because each $t\in V(T)$ has non-empty $\beta(t)\setminus \sigma(t)$ by \Cref{ob:TreeSize}. Next, (2) is because $\sum_{t\in V(T)}|\beta(t)| = \sum_{t\in V(T)}|\beta(t)\setminus \sigma(t)| + \sum_{t\in V(T)}|\sigma(t)| \leq |V(H)| + |V(T)|\cdot (2\sigma) = O(|V(H)|\sigma)$.

\begin{observation}
For each tree node $t\in V(T)$, its bag $\beta(t)$ is a strict super set of its adhesion $\sigma(t)$, i.e. $\beta(t)\supset \sigma(t)$.
\label{ob:TreeSize}
\end{observation}
\begin{proof}
To see this, in the recursion step creating $t$, we have $\sigma(t)= B$. Furthermore, $\beta(t) = V(H)$ if $X = B$ (the corner case) and $\beta(t) = X$ if $X \supset B$, so always $\beta(t)\supset \sigma(t)$. 
\end{proof}

\medskip

\noindent{\underline{The tree depth.}} We are now ready to show the depth of $T$ is $O(k\log|V(H)|/\epsilon)$. Consider a path from the root (upside) to a leaf (downside) in the recursion tree. We divide this path into two parts, the upper part and the lower part, at the upper most step $\Decomp(H',B')$ such that $|V(H')|\leq 10k\lceil 1/\epsilon\rceil$. 

For the lower part, i.e. from the step $\Decomp(H',B')$ to the leaf, trivially there are at most $|V(H')| = O(k/\epsilon)$ steps (note that the tree decomposition of $\Decomp(H',B')$ has at most $|V(H')|$ nodes).

Consider the upper part. Each step $\Decomp(H',B')$ on this part has $V(H')>10k\lceil 1/\epsilon\rceil$. The number of case-2 steps on the upper part is at most $O(\log n)$ by the following reasons. Let $\Decomp(H',B')$ be a case-2 step and $\Decomp(H'',B'')$ be its sub-step on this path. Because $X$ is $1/2$-balanced and has adhesion $2\sigma$ by \Cref{claim:DecompCase2}, we have $|V(H'')|\leq |V(H')|/2 + 2\sigma$, which implies $|V(H'')|\leq 5|V(H)|/6$ combining $2\sigma\leq |V(H)|/3$ (recall that $2\sigma = 2\lceil1/\epsilon\rceil k + 2k$ and $|V(H)|>10\lceil 1/\epsilon \rceil k$). Next, the number of \emph{consecutive} case-1 recursive steps is at most $\sigma$, because the earliest case-1 step has boundary size at most $2\sigma$ and the subsequent case-1 steps will keep reducing the boundary size by \Cref{claim:DecompCase1}, until the boundary size drops below $\sigma + 1$. Therefore, the total length of the upper part is at most $O(\log n)\cdot \sigma = O(k\log n/\epsilon)$.

In summary the recursion depth is at most $O(k\log n/\epsilon) + O(k/\epsilon) = O(k\log n/\epsilon)$.

\medskip

\noindent{\underline{Compactness and subtree unbreakability}.} Finally, the compactness and subtree unbreakability of $(T,\beta)$ simply follow from the algorithm.

\paragraph{Running Time of $\Decomp(G,\emptyset)$.}

Observe that, at each level $i$, the total graph size of level-$i$ steps is $O(m)$, because the recursive steps $\Decomp(H_{i}, B_{i})$ at a level $i$ are on edge-disjoint graphs $H_{i}$. Because there are $O(k\log n/\epsilon)$ levels, the total graph size of each level is $O(km\log n/\epsilon)$. In each recursive step $\Decomp(H,B)$, invoking \Cref{lemma:CheckUnbreakable,lemma:SepBySampleSet,thm:reducing_adhesion} takes time $2^{O(\frac{k}{\epsilon}\log\frac{k}{\epsilon})}|E(H)|^{1+o(1)+\epsilon}$. In summary, the total running time of $\Decomp(G,\emptyset)$ is $2^{O(\frac{k}{\epsilon}\log\frac{k}{\epsilon})}m^{1+o(1)+\epsilon}$.
\end{proof}

\subsection{Depth Reduction}
\label{sect:DepthReduction}

Finally, we discuss how to reduce the depth of the unbreakable decomposition from $O(k\log n/\epsilon)$ to $O(\log n)$, at a cost of bringing constant factors to the unbreakability and adhesion. Precisely, we will show how to compute a $(O(k/\epsilon),k)$-unbreakable decomposition with adhesion $O(k/\epsilon)$ and depth $O(\log n)$.

We will change \Cref{alg:Decomp} slightly as follows.
\begin{enumerate}
\item We will set $\sigma = 5\lceil 1/\epsilon \rceil k$ (instead of $\sigma = \lceil 1/k\rceil + k$).
\item Before Line 1, we add a termination condition: if $|V(H)|\leq 10\lceil 1/\epsilon\rceil k$, we return a trivial decomposition $(T,\beta)$ with a single bag $V(H)$. We point out that after adding this new termination condition, the algorithm will never go into \Cref{line:corner}, because by \Cref{claim:CornerCase}, there should be $|V(H)|\leq 2\sigma = 10\lceil 1/\epsilon\rceil k$ when we reach \Cref{line:corner}, which is impossible because of the new termination condition.
\item In \Cref{line:CheckUnbreakable}, we will check if $B$ is $(\lceil 1/\epsilon\rceil k + k, k)$-unbreakable (instead of $(k,k)$-unbreakable). This will lead to a new \Cref{claim:DecompCase1}:
\begin{itemize}
\item In Case 1, the set $X$ is $(O(k/\epsilon),k)$-unbreakable and has adhesion at most $|B| - \lceil 1/\epsilon \rceil k$.
\end{itemize}
The new claim can be shown in the same way as \Cref{claim:DecompCase1}. Roughly speaking, in case (a), $X$ is $(k+2\sigma,k)$-unbreakable and has adhesion at most $|B| - \lceil 1/\epsilon\rceil k$. In case (b), $X$ is $(2k\lceil 1/\epsilon\rceil + k, k)$-unbreakable and has adhesion $2k\lceil 1/\epsilon\rceil + k\leq \sigma\leq |B| - \lceil 1/\epsilon\rceil k$.
\end{enumerate}

Following the same argument in the proof of \Cref{thm:DetailedMain}, it is not hard to see that the resulting decomposition is $(O(k/\epsilon),k)$-unbreakable and has $O(k/\epsilon)$ adhesion. It also has all the other properties stated in \Cref{thm:DetailedMain}. In particular, the depth reduces to $O(\log n)$ by the following reasons. Consider a path from the root to an arbitrary leaf. The number of case-2 steps is at most $O(\log n)$ (note that there is no lower part because of the new termination condition). The number of consecutive case-2 steps is at most $\sigma/(\lceil 1/\epsilon\rceil k) = O(1)$, because every case-2 step will reduce the adhesion by $\lceil 1/\epsilon\rceil k$ (from modification 3).

\section{Application: \textsc{Minimum $p$-Way Cut}  in Close-to-Linear Time} 
\label{sec:app}

The goal of this section is to prove~\Cref{thm:kcut}.
We essentially use the same algorithm as that for {{\sc Auxiliary Multicut}} - Theorem 4.1 in~\cite{cygan2020randomized}, with some minor modifications. We remark that we give the full algorithm only to show that the algorithm indeed runs in close-to-linear time - indeed, the algorithm in~\cite{cygan2020randomized} itself does run in close-to-linear time modulo the construction of the decomposition, though this is not analyzed explicitly. Our algorithm is simpler to state since our unbreakable decomposition admits the subtree unbreakability property (see~\Cref{def:UnbreakableDecomp}).

Given a graph $G$, parameters $p$ and $k$, we need to decide if there exists a {\textsc{$p$-Way Cut}} of size at most $k$. Recall that a {\textsc{$p$-Way Cut}}  is a set of edges whose deletion creates $p$ connected components. We begin by computing a $(q,k)$ unbreakable rooted tree decomposition with adhesion at most $\sigma$ where $\sigma \leq q = \OO(\frac{k}{\epsilon})$ using~\Cref{thm:main}, where we set $\epsilon$ to be some small constant. Let $r$ be the root of this rooted tree decomposition. The running time for this step is $2^{\OO(\frac{k}{\epsilon}\log \frac{k}{\epsilon})}m^{1 + \epsilon}$.

\paragraph{Equivalent formulation.}Instead of checking if there exists a {\textsc{$p$-Way Cut}} of size at most $k$, equivalently, we can check if there is some $p$-coloring of the vertex set such that the number of edges with differently colored endpoints is at most $k$. Formally, in a graph $H$, a $p$-coloring of the vertex set is some function $h: V(H) \rightarrow [p]$. The cost of the $p$-coloring $cost(h)$ is the number of edges with differently colored endpoints. Given two functions $f:A \rightarrow C$ and $g:B \rightarrow C$ where $A \subseteq B$, we say that $g$ respects $f$ if $g$ agrees with $f$ on every $a \in A$. Recall~(see  \Cref{sec:preliminaries}) that given a node $t$ in the rooted tree in the unbreakable decomposition, $G_t$ is defined as $G[\gamma(t)] \setminus E(G[\sigma(t)])$. In words, $G_t$ is the graph induced on all vertices contained in bags of the sub-tree rooted at $t$, excluding the edges between the vertices in the adhesion $\sigma(t)$.

\paragraph{Dynamic programming on tree decomposition.} We use dynamic programming to solve the problem. For each tree node $t$ in the (unbreakable) tree decomposition, every $I \subseteq [p]$ and every $p$-coloring $f:\sigma(t) \rightarrow [p]$ of the adhesion $\sigma(t)$, we store an entry $M[t,f,I]$ in our table $M$. $M[t,f,I]$ is an integer value that denotes the cost of the minimum cost $p$-coloring $h: V(G_t) \rightarrow [p]$ in $G_t$ that  (a) respects $f$ (b) colors some vertex of $G_t$ with color $c$ for every $c \in I$ an (c) has cost at most $k$ in $G_t$. If there is no such coloring $h$, then we set $M[t,f,I] = \infty$. Then since for the root $r$, $G_r = G$ and $\sigma(r) = \emptyset$, $G$ has a {\textsc{$p$-Way Cut}} of size at most $k$ if and only if $M[r,f_\phi,[p]] \leq k$, where $f_\phi$ is an empty function (note that any coloring $h$ respects the empty function).

We now explain how to compute each entry $M[t,f,I]$. Since we compute these values bottom up with respect to the tree decomposition, we will inductively assume that we have already correctly computed the set of values $M[t',(.),(.)]$ for every child $t'$ of $t$ in the tree. Henceforth we assume that $M[t,f,I] < \infty$, if not, our algorithm will discover this and set $M[t,f,I] = \infty$. Fix a minimum cost $p$-coloring $h$ satisfying the properties listed above, such that $M[t,f,I]$ = cost($h$). Consider a restriction of this $p$-coloring $h$ to $\beta(t)$, the bag corresponding to node $t$. Let $g:\beta(t) \rightarrow [p]$ be this coloring.

\paragraph{Exploiting unbreakability.}The first observation is that since $\beta(t)$ is $(q,k)$ unbreakable in $G_t$, the total number of vertices of $\beta(t)$ in all except one colored class of $g$ is at most $3q$ - for otherwise, we would violate the definition of unbreakability. We show this formally in the next lemma.

\begin{lemma}\label{lemma:colorbound}
There exists at most one $c \in [p]$ such that there exist more than $q$ vertices $v \in \beta(t)$ for which $g(v) = c$. Further, the total number of vertices which have color $c' \neq c$ is at most $3q$.  
\end{lemma}
\begin{proof}
To prove the first part, suppose for contradiction that there exist $c_1,c_2 \in [p]$ with $c_1 \neq c_2$ such that there exist sets $D_1 \subseteq \beta(t)$ and $D_2 \subseteq \beta(t)$ with $g(v_1) = c_1$ for any $v_1 \in D_1$ and $g(v_2) = c_2$ for any $v_2 \in D_2$ with both $|D_1|,|D_2| \geq q + 1$.

Recall that we assumed that there exists a $p$-coloring $h$ of $G_t$ with cost at most $k$, whose restriction to $\beta(t)$ gives $g$. Let $S_1 = h^{-1}(c_1)$ and $S_2 = h^{-1}(c_2)$, where $h^{-1}(c_i)$ is the set of vertices $v \in V(G_t)$ for which $h(v) = c_i$, $i \in \{1,2\}$. Let $S_{12}$ be the set of vertices in $h^{-1}(c_1)$ which have some neighbor in $h^{-1}(c_2)$ in the graph $G_t$. Since the cost of $h$ is at most $k$, we must have $|S_{12}| \leq k$.   
Now consider the vertex cut $(L = S_1, R = S_{12} \cup S_2)$ in $G_t$. Indeed, $(L,R)$ is a vertex cut since $L \cap R = S_{12}$ separates $S_1 \setminus S_{12}$ from $S_2$ in $G_t$. Also we have $|L \cap R| = |S_{12}| \leq k$. Finally, $g$ is a restriction of $h$ so it follows that $D_1 \subseteq L$ and $D_2 \subseteq R$, and both $|D_1|, |D_2| \geq q + 1$. Hence $(L,R)$ is a $(q,k)$ breakable witness for $\beta(t)$ in $G_t$, which is a contradiction since $\beta(t)$ is $(q,k)$ unbreakable in $G_t$ (recall that our unbreakable decomposition admits subtree unbreakability, see~\Cref{def:UnbreakableDecomp}).%
Thus every colored class except one, say $c$, has at most $q$ vertices. 

Next, we show that the total number of vertices in $\beta(t)$ colored differently from $c$ is at most $3q$. If this is not the case, since $|h^{-1}(c') \cap \beta(t)| \leq q$ for each $c' \neq c$, we first observe that we can partition the set of colors $[p] \setminus \{c\}$ into two (disjoint) groups $X_1 \cup X_2$, such that there are at least $q + 1$ vertices of $\beta(t)$ having colors from $X_1$ and $X_2$.

For a set of colors $X$, let us define $h^{-1}(X)$ to be the set of vertices $v \in V(G_t)$ such that $h(v) \in X$. Then it follows that $|h^{-1}(X_1)|, |h^{-1}(X_2)| \geq q + 1$. Similar to the previous part, let us define $S_1' = h^{-1}(X_1)$, $S_2' = h^{-1}(X_2)$, and $S_{12}'$ to be the set of vertices in $S_1'$ that have a neighbor in $S_2'$. Again, it must be the case that $(L = S_1', R = S_{12}' \cup S_2')$ is a $(q,k)$ breakable witness for $\beta(t)$ in $G_t$, which is a contradiction.
\end{proof}

We guess the label of the color class $c \in [p]$ which has the maximum number of vertices $v \in \beta(t)$ satisfying $g(v) = c$, then we can potentially have more than $q$ vertices of that color, but not for any other color. Without loss of generality and for simplicity of exposition let us assume that this color $c$ is the color $p$.

\paragraph{Computing the table entries.}To compute $M[t,f, I]$, our goal will be to come up with a coloring $g^*: \beta(t) \rightarrow [p]$, and a set $I_{t'} \subseteq I$ for each child $t'$ of $t$ in the tree such that the following hold.

\begin{enumerate}
\item \label{item:opt1} $g^*$ respects $f$ 
\item \label{item:opt2} For every $i \in I$, either $i \in I_{t'}$ for some child $t'$ of $t$, or there exists some $v \in \beta(t)$ such that $g^*(v) = i$.

\end{enumerate}

Define the cost of $g^*$ as follows. Every edge $e = \{u,v\} \in E(\beta(t)) \setminus E(\sigma(t))$ contributes $1$ to the cost. For every child $t'$ of $t$, $t'$ contributes $M[t',\left.g^*\right|_{\sigma_{t'}},I_{t'}]$ to the cost, where $\left.g^*\right|_{\sigma_{t'}}$ denotes the restriction of $g^*$ to the set $\sigma(t')$. Then observe that since we correctly computed the values $M[t', (.),(.)]$ for each child $t'$ of $t$, such a coloring $g^*$ with minimum cost must have cost equal to $M[t,f,I]$.

Notice that in particular, $g$ itself is a coloring of minimum cost that satisfies these conditions. To find a coloring of minimum cost (possibly different from $g$), we use color coding together with another dynamic programming step. We start by describing the color coding step.

\paragraph{Color coding.}For every edge $e = \{u,v\}$ in $E(\beta(t)) \setminus E(\sigma(t))$, we say that $e$ is \emph{crossing} if $g(u) \neq g(v)$. Similarly, for every child node $t'$ of $t$ in the tree decomposition tree and adhesion $\sigma(t') \subseteq \beta(t)$, we say that $\sigma(t')$ is \emph{crossing} if there exists two vertices $u,v \in \sigma(t')$ such that $g(u) \neq g(v)$.

\begin{lemma}
The total number of crossing edges and crossing adhesions together is at most $k$.
\end{lemma}
\begin{proof}
The proof follows from the compactness (see \Cref{sec:preliminaries}) property of the decomposition. Compactness implies that for each child $t'$ of $t$, $G[\alpha(t')] = G[\gamma(t') \setminus \sigma(t')]$ is connected and $N_{G_{t'}}(\alpha(t')) = \sigma(t')$. Also, the graphs $G_{t'} = G[\gamma(t')] \setminus E(\sigma(t'))$ are disjoint across all children $t'$ of $t$. If $\sigma(t')$ is crossing, then there must exist some two vertices $u,v \in \sigma(t')$ such that $g(u) \neq g(v)$, and hence each crossing adhesion contributes $1$ to the cost of $h$.

Similarly, every crossing edge $e \in E(\beta(t)) \setminus E(\sigma(t))$ contributes $1$ to the cost of $h$. But the cost of $h$ in $G_t$ is at most $k$, thus the total number of crossing edges and adhesions together is at most $k$.
\end{proof}

Let $A_1, A_2 \ldots A_p$ be the set of all vertices of $\beta(t)$ colored $1,2 \ldots p$ respectively in $g$. Let $A^* = \bigcup_{i = 1}^{p-1} A_i$ be the set of vertices colored with colors $1,2 \ldots p-1$. Define $B^* \subseteq A_p$ as the set of all vertices $v \in A_p$ for which there is either (a) an edge $\{u,v\} \in E(\beta(t)) \setminus E(\sigma(t))$ with $u \in A^*$ or (b) some child node $t'$ of $t$ such that $\sigma(t')$ is crossing, and $v \in \sigma(t')$.
 
Apply~\Cref{lemma:hitmiss} to the sets $A_1, A_2 \ldots A_{p-1}, B^*$ to obtain a family of functions $\mathcal{F}$. The running time is at most $2^{\OO(q \log q)} n \log^2 n$, since $\sum_{i = 1}^{k-1} A_i \leq 3q$ by~\Cref{lemma:colorbound} and $|B^*| \leq qk\sigma \leq q^3$, where $\sigma$ is the adhesion of the unbreakable decomposition. Then we know that there exists some function $g' \in \mathcal{F}$ that agrees with $g$ on these sets. Henceforth we fix this $g'$ (we try every function from $\mathcal{F}$). Before we proceed further, for every $v \in \sigma(t)$, we make $g'$ agree with $f$ by setting $g'(v) =  f(v)$. Let us also assume that at least one vertex of $\beta(t)$ is colored $p$ in $g'$ (note that this is true for $g$ since $p$ is the color with maximum frequency; hence this vertex can be included in the color coding analysis to ensure this, but we avoid explicity stating this to avoid confusion). Now we set $I = I \setminus \{p\}$.

Next, consider the graph $H$ obtained as follows. The vertex set $V(H) = \beta(t)$. We add into $E(H)$ every edge of $E(G_t) \setminus E(\sigma_t)$.    Further for every child node $t'$ of $t$, we turn the adhesion $\sigma(t')$ into a clique. Formally, for every $u ,v \in \sigma(t')$ with $u \neq v$ we add the edge $\{u,v\}$ to $E(H)$.

Let $H^*$ be the subgraph of $H$ induced on the vertices which are colored $\{1,2 \ldots p-1\}$ by $g'$, and let $\{C_1, C_2 \ldots C_{\ell}\}$ be the connected components of $H^*$. We have the following observation.

\begin{observation}
For connected component $C_i$, $i \in [\ell]$, either $g'$ agrees with $g$ on vertices of $C_i$ or $g(v) = p$ for each $v \in C_i$. 
\end{observation}

\begin{proof}
Note that if an edge $\{u,w\} \in E(H)$ is such that $g(u) \neq g(w)$, then by our color coding requirements, $g'(u) = g(u)$ and $g'(w) = g(w)$. Also if $g(u),g(w) \in [p-1]$, then $g'(u) = g(u)$ and $g'(w) = g(w)$.

Consider two cases. The first case is when in $g$, each vertex of $C_i$ is colored  from the set of colors $[p-1]$. In this case, $g$ and $g'$ agree on every vertex of $C_i$.

Otherwise, there exists a vertex of $C_i$ that is colored with color $p$ in $g$. We have two subcases now. Either every vertex of $C_i$ is colored $p$, or there is some edge $\{u,w\} \in E(H)$ such that $g(u) \in [p-1]$ and $g(w) = p$. But in the latter case, by construction, we have $w \in B^*$, and hence $g'(w) = g(w) = p$, which is a contradiction since every vertex of $C_i$ is colored from the set $[p-1]$ in $g'$.
\end{proof}

Thus there exists an optimal coloring (the coloring $g$), which for each connected component $C_i$, $i \in [\ell]$, either colors $C_i$ consistently with $g'$ or colors every vertex of $C_i$ with the color $p$. Equivalently, in order to compute a coloring $g^*$ of minimum cost satisfying~\cref{item:opt1,item:opt2}, we start with the coloring $g^* = g''$ which colors every vertex in $\beta(t)$ with the color $p$, and then decide to \emph{flip} some connected components $C_i$, $i \in [\ell]$, so that we change the colors of vertices of $C_i$ in $g^*$ to color them consistently with $g'$.

For each component $C_i$, $i \in [\ell]$, define the \emph{cost of flipping} as $$\text{flip}(C_i) = \sum_{\{u,v\} \in E(\beta(t)) \setminus E(\sigma(t))
} \mathbb{I}({u,v \in C_i\;\text{and}\;g'(u) \neq g'(v)}) + \sum_{\{u,v\} \in E(\beta(t)) \setminus E(\sigma(t))} \mathbb{I}(u \in C_i, v \notin C_i).$$ 

Essentially this is the cost for flipping $C_i$: we charge for all edges in the bag $\beta(t)$ with differently colored endpoints inside $C_i$, and for all edges which have exactly one endpoint in $C_i$ (note that in this case the other endpoint is colored $p$). Before we proceed further, we note that if a component $C_i$ contains a vertex $v \in C_i \cap \sigma(t)$, then since $g'(v) = f(v) = g(v) \neq p$, this already fixes the decision for this component: we flip $C_i$.

\paragraph{Flipping components using dynamic programming.}It now suffices to decide which subset of components to flip. To accomplish this we will yet again use dynamic programming. For this, we need a few more observations and definitions.

First, observe that every adhesion of a child node is a clique in $H$, and hence can intersect at most one $C_i$, $i \in [\ell]$. For each component $C_i$, let $x_i$ be the number of children $t'$ of the node $t$ whose adhesion $\sigma(t')$ intersects  $C_i$. Arbitrarily order these child nodes by numbering them $(i,1)$, $(i,2)$, $(i,3)$ \ldots $(i,x_i)$, so that $t'_{i',j'}$ refers to the child node labelled $(i',j')$. Thus $\sigma(t'_{i',j'})$ refers to the adhesion of the child node labeled $(i',j')$. For simplicity, let us assume that for every child node $t'$ of $t$, $\sigma(t')$ intersects some $C_i$, $i \in [\ell]$ - if not, this is easy to handle using a simple modification to the dynamic programming.

Given a coloring $\psi: \beta(t) \rightarrow [p]$, $i \in [\ell]$ and $j \in [x_i]$, a \emph{flip set} $F \subseteq [i]$ and assignments $I_{t'_{i',j'}} \subseteq I$ to each child node numbered $(i',j')$ satisfying either (a) $i' < i$, or (b) $i' = i$ and $j' \leq j$, the \emph{$i,j$-partial cost} of $\psi$ is defined as follows:

\begin{itemize}
\item Each component $C_{i'}$, such that $i' \in F$, contributes a cost of $\text{flip}(C_i)$.
\item Each child node $t'_{i',j'}$ such that either (a) $i' < i$, or (b) $i' = i$ and $j' \leq j$ contributes a cost $M[t', \left.\psi\right|_{\sigma_{t'_{i',j'}}}, I_{t'_{i',j'}}]$ where $\left.\psi\right|_{\sigma_{t'_{i',j'}}}$ is the restriction of $\psi$ to $\sigma(t'_{i',j'})$.

\end{itemize}

 We create a dynamic programming table, whose each entry is of the form $T[i][j][I_0][b]$, $i \in \{0, 1, 2 \ldots \ell\}, j \in \{0,1, \ldots x_i\}, I_0 \subseteq I, b \in \{0,1\}$. $T[i][j][I_0][b]$ denotes the minimum possible $i,j$-partial cost across all choices of colorings $\psi: \beta(t) \rightarrow [p]$, flip sets $F \subseteq [i]$ and sets $I_{t'_{i',j'}} \subseteq I_0$ for all $(i',j')$ such that either $i' < i$ or $i' = i, j' \leq j$, while maintaining the following properties.

\begin{itemize}
\item  $\psi$ is obtained from $g''$ by flipping some subset of components with indices equal to $F \subseteq [i]$.
\item If $b = 1$, then $i \in F$, else $i \notin F$.
\item For each $i' \leq i$, if $C_{i'}
\cap \sigma(t) \neq \emptyset$, then $i' \in F$.

\item For every $i_0 \in I_0$, either there exists some vertex $v \in \beta(t)$ with $\psi(v) = i_0$, or there exists a child node of $t$ labelled $i',j'$ with either (a) $i' <i$ or (b) $i' = i$ and $j'\leq j$ such that $v \in I_{t'_{i',j'}}$.
\end{itemize}

The updates to $T[i][j][I_0][b]$ are natural. For the base cases, we set $T[0][0][\emptyset][0] = T[0][0][\emptyset][1] = 0$. For any other set $I_0$ and $b' \in \{0,1\}$ we set $T[0][0][I_0][b'] = \infty$. If $i > 0$ and $j = 0$, there are two cases. If $C_i \cap \sigma(t) \neq \emptyset$, then we set $T[i][j][I_0][0] = \infty$. Else we just set $T[i][j][I_0][0] = \max\{T[i-1][x_{i-1}][I_0][0], T[i-1][x_{i-1}][I_0][1]\}$. In either case we set $T[i][j][I_0][1] = \max\{T[i-1][x_{i-1}][I_0 \setminus I_i][0], T[i-1][x_{i-1}][I_0 \setminus I_i][1]\} + \text{flip}(C_i)$, where $I_i$ is the set of colors given to vertices in $C_i$ in $g'$.

For $j > 0$, if $\sigma(t) \cap C_i \neq \emptyset$ we set $T[i][j][I_0][0] = \infty$. Otherwise, we set 
$$T[i][j][I_0][0] = \min_{I_0'' \subseteq I_0} M[t'_{{i,j}}][f_0][I_0''] + T[i][j-1][I_0 \setminus I_0''][0].$$ 
In either case, we set

$$T[i][j][I_0][1] = \min_{I_0'' \subseteq I_0} M[t'_{i,j}][f_1][I_0''] + T[i][j-1][I_0 \setminus I_0''][1].$$

Here $f_0$ is the constant function that assigns to every vertex of $\sigma(t'_{i,j})$ the color $p$, and $f_1$ is the function that colors each vertex of $\sigma(t'_{i,j})$ the same color as that in $g'$.

Finally, we set $M[t,f,I]$ to be $\min(T[\ell][x_{\ell}][I_0][0], T[\ell][x_{\ell}][I_0][1]\}$. Recall that in reality, we repeat this process for every function $g' \in \mathcal{F}$ from the color coding step: therefore we define $M[t,f,I]$ to be the minimum over this result for every function in $\mathcal{F}$. If this quantity is more than $k$ for every choice of the coloring from $\mathcal{F}$ obtained, we just set $M[t,f,I]$ to $\infty$.

\paragraph{Running time.} The running time analysis is rather straightforward. Let us first analyse the run-time given the unbreakable decomposition. Fix an entry $M[t,f,I]$. There are at most $p^{\OO(\sigma)}$ such entries for every tree node $t$. To compute each entry, we first do a color coding step and a dynamic programming step which constructs the table $T$. Let $x_t$ be the total number of children of $t$ in the tree decomposition. Let $n_{tv} = |\beta(t)|$ be the number of vertices in the bag and $n_{te}$ be the number of edges in $E(\beta(t)) \setminus E(\sigma(t))$.%

The color coding step takes time $\OO(2^{\OO(\frac{k}{\epsilon} \log \frac{k}{\epsilon})} (n_{tv} + n_{te}) \log^2 n)$.

There are at most $(x_t + n_{tv}) \cdot 2^{p} \cdot 2 = \OO((x_t + n_{tv}) \cdot 2^{p})$ entries in the table $T$. Each entry is computed in time $\OO(2^{p})$. Thus the total running time to compute $T$ is $\OO((x_t + n_{tv})2^{p})$. Summing across all nodes $t$, using the facts that $\sum_t (n_{tv} + n_{te}) = \OO(m + kn/\epsilon)$ ($\sum_{t}n_{tv}\leq \OO(kn/\epsilon)$ is from \Cref{thm:DetailedMain} and $\sum_{t} n_{te} = \OO(m)$ is from the nature of tree decomposition), and further noting that $\sum_t x_t = \OO(n)$ by~\Cref{thm:DetailedMain}, the total running time is at most $2^{\OO(\frac{k}{\epsilon} \log \frac{k}{\epsilon})} (m + n) \log^2 n$.

Finally, the unbreakable decomposition itself is computed using~\Cref{thm:DetailedMain}, and this is the bottleneck for the running time. The running time of the entire algorithm is therefore $2^{\OO(\frac{k}{\epsilon} \log \frac{k}{\epsilon})}m^{1 + \epsilon}$. This concludes the run-time analysis.

\section{Conclusion and Open Problems}
\label{sec:conclusion}
\Cref{thm:DetailedMain} gives the first close-to-linear time FPT algorithm for unbreakable decomposition. Our decomposition also has optimal unbreakability and adhesion parameters up to a constant factor. This removes the bottleneck to fast FPT algorithms for numerous problems. 

Below, we list some exciting potential applications that relied on unbreakable decomposition:

\begin{enumerate}

\item Can we find a {\textsc{Minimum Bisection}} of size $k$ in $\OO_{k}(m^{1+\epsilon})$ time, if it exists?

\item Can we improve the preprocessing time of the $\OO(1)$-update-time connectivity oracle under $\OO(1)$ vertex failures from \cite{pilipczuk2021algorithms} to close-to-linear time? What is the optimal trade-off for space and preprocessing time when the update time is $\OO(1)$?

\item Let $G$ be a graph that excludes a fixed topological minor. Can we decide whether $G$ satisfies an $\mathsf{FO+conn}$ sentence or an $\mathsf{FO+DP}$ sentence in $\OO(m^{1+\epsilon})$ time? This would be an improvement over the $\OO(n^{3})$ bound given by \cite{pilipczuk2021algorithms,schirrmacher2024model}. It is worth noting that there are already a linear time algorithm for deciding an $\mathsf{MSO}_{2}$ sentence in bounded-treewidth graphs \cite{courcelle1990monadic}, and a close-to-linear time algorithm for deciding an $\mathsf{FO}$ sentence in nowhere dense graphs \cite{grohe2017deciding}.
\end{enumerate}

Can our running time be improved further? More specifically, can one compute $(O(k),k)$-unbreakable decomposition with adhesion $O(k)$ in $O_k(m\log^{O(1)} m)$ or $O_k(m^{1+o(1)})$) time? 

Moreover, we can hope to remove the exponential dependency on $k$ in the edge-cut version. 
As shown in \cite{lokshtanov2022parameterized}, a $(\poly(k),k)$-\emph{edge}-unbreakable decomposition can be computed in polynomial time with no exponential dependency on $k$. 
A near-linear time algorithm for computing an $(\poly(k \log n),k)$-edge-unbreakable decomposition will likely imply a nice application, i.e., a near-linear time $(1+\epsilon)$-approximation FPT algorithm for \textsc{Minimum $p$-Way Edge-Cut} parameterized by $p$, improving the polynomial running time of the algorithm in \cite{lokshtanov2022parameterized}.

\bibliographystyle{alpha}
\bibliography{references}

\appendix
\section{Single Source Vertex Mincuts}\label{sec:singlesource}

In this section, we show how to compute single source vertex mincuts. Our proof is based on previous algorithms for computing min-cut thresholds in graphs. Given a graph $G$, a single source $s \in V(G)$, a set of targets $T \subseteq V(G)$ and a parameter $\mu$, the single source min-cut threshold problem asks to find the set of all targets $t \in T$  such that the size of the $(t,s)$ edge/vertex min-cut is at most $\lambda$.~\cite{li2021approximate} showed an algorithm for the edge version using minimum isolating cuts~\cite{li2020deterministic,abboud2021subcubic}. Later,~\cite{pettie2022optimal} generalized this to the vertex version using minimum isolating vertex cuts~\cite{li2021vertex,chekuri2021isolating}. Their algorithm works on unit-capacity graphs, but it is easy to see that this generalizes to capacitated graphs. For our setting, we need a slightly different notion in that we need exact min-cuts instead of approximate min-cuts: but this can be ensured easily by running the threshold algorithm for each parameter $\mu$ from $1$ to $k$ by paying a factor of $k$ in the running time, which we can indeed afford. For completeness, we give the complete proof for our setting, and we follow the simpler proof shown in \cite{li2023near}.

\thmssmc*

We will use the isolating vertex cuts subroutine stated in \Cref{thm:maxiso}. Recall that a ${\cal C}$ is collection of disjoint cuts if the cuts $(L,R)\in{\cal C}$ have mutually disjoint $L\setminus R$.

\begin{definition}[Isolating Cuts]
Consider a capacitated graph $\wtilde{G}$ with terminals $W \subseteq V(\wtilde{G})$ that form an independent set. For each terminal $w \in W$, an \emph{isolating cut} of $w$ with respect to $W$ is an arbitrary $w$-$(W \setminus \{w\})$ mincut.
\end{definition}

\begin{theorem}[Isolating Vertex Cuts \cite{li2021vertex,chekuri2021isolating}]\label{thm:maxiso}
Given a capacitated graph $\wtilde{G}$ with terminals $W\subseteq V(\wtilde{G})$ that form an independent set, there is an algorithm that computes a collection ${\cal C}$ of disjoint cuts satisfying that
\begin{itemize}
\item for each $w\in W$, there is a $(L,R)\in{\cal C}$ which is an isolating cut of $w$;
\item each $(L,R)\in{\cal C}$ is an isolating cut of some $w\in W$.
\end{itemize}
The running time is $m^{1+o(1)}$.
\end{theorem}

We note that the disjointness of the isolating cuts from \Cref{thm:maxiso} can be seen easily following the algorithm in \cite{li2021vertex}.

\begin{algorithm}[H]
\caption{Single Source Vertex mincut}
\label{algo:SSMC}
\begin{algorithmic}[1]

\Require a capacitated graph $\wtilde{G}$ with a parameter $k$, a source vertex $s$ and sink vertices $T$.
\Ensure A set $\CC$ of vertex cuts. \
\State Initialize $\Gamma$ to be an empty vertex set and ${\cal K}$ to be an empty mincut cover.
\For {$k'=1,2,...,k$}
\For {$j$ = $1,2 \ldots O(\log^2 n)$}
\State Initialize ${\cal K}'$ to be an empty mincut cover.
\For {$i = 0, 1, 2 \ldots \lfloor \log n \rfloor,r=2^{i}$}

\State\label{line:sample}Sample each sink vertex in $T \setminus \Gamma$ with probability $\frac{1}{r}$ independently into a set $T'$.
\State $W\gets T' \cup \{s\}$
\State\label{line:IsolatingCuts}Let ${\cal C}_{\iso}$ be the collection of isolating cuts of $W$ in $\wtilde{G}$ (use \Cref{thm:maxiso})
\State\label{line:CutCollection}Let ${\cal C}$ collect the cuts $(L,R)\in {\cal C}_{\iso}$ which has size $k'$ and is an isolating cut of some sink $t\in T'\subseteq W$
\State Add ${\cal C}$ into ${\cal K}'$.
\EndFor
\State\label{line:UpdateGamma}For each cut $(L,R)\in{\cal K}'$, $\Gamma\gets \Gamma\cup (L\setminus R)\cap (T\setminus \Gamma)$
\State ${\cal K}\gets {\cal K}\cup {\cal K}'$
\EndFor

\EndFor
\State \Return ${\cal K}$

\end{algorithmic}
\end{algorithm}

\begin{proof}[Proof of \Cref{thm:SSMC}]

The single source vertex mincut algorithm is described in \Cref{algo:SSMC} in details. In what follows, we prove that ${\cal K}$ is a mincut cover with respect to $s$ and $T$.

Instead of directly verifying ${\cal K}$ satisfies the properties in \Cref{def:MincutCovers}, we will first show \Cref{lemma:Induction}.

\begin{lemma}
With high probability, at the end of each phase $k'$ of the outer loop, $\Gamma$ exactly collects all sink vertices $t$ with $\lambda_{\wtilde{G}}(t,s)\leq k'$, i.e., at the beginning of each phase $k'$, $T\setminus \Gamma$ exactly contains sink vertices with $\lambda_{\wtilde{G}}(t,s)\geq k'$. 
\label{lemma:Induction}
\end{lemma}
\begin{proof}
We will use induction.
Note that initially this claim holds at the beginning of the first phase $k'=1$.
To enable the induction, fix an arbitrary phase $1\leq k'\leq k$, and we will show that this claim holds at the end of phase $k'$ with high probability, assuming that it holds at the beginning of phase $k'$. 

For each single iteration $j$ of the middle loop, let $\wtilde{T}_{j}$ collect all sink vertices $t\in T\setminus \Gamma$  s.t. $\lambda(t,s) = k'$ at the beginning of iteration $j$, and let $\ell_{j}$ be a random variable denoting the size of $\wtilde{T}_{j}$. \Cref{claim:PotentialDrop} says that $E[\ell_{j}]$ will drop by a factor of $\Omega(1/\log n)$ in each iteration $j$ of the middle loop. Therefore, after $O(\log^{2}n)$ (with sufficiently large hidden constant) iterations $j$, in expectation the number of sink vertices $t\in T\setminus \Gamma$ s.t. $\lambda(t,s)=k'$ is at most $1/\poly(n)$. By Markov's inequality, with high probability, the number of sink vertices $t\in T\setminus \Gamma$ s.t. $\lambda(t,s) = k'$ at the end of phase $k'$ is $0$.

\begin{claim}
For each iteration $j$, we have $E[\ell_{j+1}]\leq \ell_{j}(1-\Omega(1/\log n))$.
\label{claim:PotentialDrop}
\end{claim}
\begin{proof}
At the end of each iteration $j$, we define a random variable
\[
\alpha_{j} = \sum_{{\cal C}\in{\cal K}'}\sum_{(L,R)\in {\cal C}} |(L\setminus R)\cap (T\setminus \Gamma)|
\]
We will see in a moment that $E[\alpha_{j}]\geq \Omega(\ell_{j})$. Providing this, the number of new sink vertices added to $\Gamma$ in \Cref{line:UpdateGamma} at the end of iteration $j$ is at least $\Omega(\ell_{j}/\log n)$ in expectation, because each vertex in $T\setminus \Gamma$ will be counted at most $O(\log n)$ times in $\alpha_{j}$ (at most once in each collection ${\cal C}$ because \Cref{thm:maxiso} return a collection of disjoint cuts). Note that these new sink vertices added into $\Gamma$ are all $\wtilde{T}_{j}$-vertices (since all cuts in ${\cal K'}$ have size $k'$), we conclude that $E[\ell_{j+1}]\leq \ell_{j}-\Omega(\ell_{j}/\log n)$. 

It remains to show $E[\alpha_{j}]\geq \Omega(\ell_{j})$. Fix a sink vertex $t\in \wtilde{T}_{j}$, and let $(L^{*},R^{*})$ denote the $t$-$s$ mincut that minimizes $|L^{*}\setminus R^{*}|$. By the submodularity of vertex cuts, we have \Cref{ob:Closest}.

\begin{observation}
Any $t$-$s$ mincut $(L,R)$ has $L\setminus R\supseteq L^{*}\setminus R^{*}$
\label{ob:Closest}
\end{observation}

We consider the size scale $r_{t} = 2^{\lfloor |(L^{*}\setminus R^{*})\cap \wtilde{T}_{j}| \rfloor}$ (i.e. $r_{t}$ is roughly the size of $(L^{*}\setminus R^{*})\cap \wtilde{T}_{j}$). Regarding the sample process in \Cref{line:sample} of this size scale $r_{t}$, we define an event $Q_{t}$ in which $t$ is sampled but the other sink vertices in $(L^{*}\setminus R^{*})\cap \wtilde{T}_{j}$ are not. Trivially $Q_{t}$ happens with probability 
\[
\frac{1}{r_{t}}\cdot (1-\frac{1}{r_{t}})^{|(L^{*}\setminus R^{*})\cap \wtilde{T}_{j}|-1} = \Omega(1/r_{t}) = \Omega(1/|(L^{*}\setminus R^{*})\cap \wtilde{T}_{j}|).
\]

Note that $(L^{*}\setminus R^{*})\cap (T\setminus \Gamma) = (L^{*}\setminus R^{*})\cap \wtilde{T}_{j}$ because each $t\in (L^{*}\setminus R^{*})\cap (T\setminus \Gamma)$ has $\lambda_{\wtilde{G}}(t,s)$ by the induction hypothesis and the fact that $(L^{*},R^{*})$ has size $k'$. %
Providing this, whenever $Q_{t}$ happens, we must have $(L^{*}\setminus R^{*})\cap T' = \{t\}$, which means $(L^{*},R^{*})$ is a $t$-$(W\setminus\{t\})$ cut. Therefore, the isolating cut $(L,R)\in{\cal C}_{\iso}$ of $t\in T'\subseteq W$ computed in \Cref{line:IsolatingCuts} has size exactly $k'$, which means in \Cref{line:CutCollection}, $(L,R)$ will be added into ${\cal C}'$. By \Cref{ob:Closest}, 
\[
|(L\setminus R)\cap (T\setminus \Gamma)| \geq |(L^{*}\setminus R^{*})\cap (T\setminus \Gamma)| = |(L^{*}\setminus R^{*})\cap \wtilde{T}_{j}|.
\]
In summary, with probability $\Omega(1/|(L^{*}\setminus R^{*})\cap \wtilde{T}_{j}|)$, the cut $(L,R)$ will contribute $|(L\setminus R)\cap (T\setminus \Gamma)|\geq |(L^{*}\setminus R^{*})\cap \wtilde{T}_{j}|$ to $\alpha_{j}$. In other words, $(L,R)$ contributes $\Omega(1)$ to $E[\alpha_{j}]$. Repeating this argument for all $t\in \wtilde{T}_{j}$, we can conclude that $E[\alpha_{j}]\geq \Omega(|\wtilde{T}_{j}|) = \Omega(\ell_{j})$.
\end{proof}
\end{proof}

We are ready to verify that ${\cal K}$ is a mincut cover with respect to $s$ and $T$.

\begin{itemize}
\item To prove property \ref{prop:Cover1}, we argue that at phase $k'$, each cut $(L,R)$ in each collection ${\cal C}$ (obtained in \Cref{line:CutCollection}) is a $t$-$s$ mincut for some $t\in T^{*}$.

Recall that the cut $(L,R)$ we add into ${\cal C}$ has size $k'$ and is an isolating cut of some sink $t\in T'\subseteq W$. Observe that $t\in T'\subseteq T\setminus \Gamma$ from \Cref{line:sample}. Combining it with \Cref{lemma:Induction} and the fact that $\Gamma$ only grows, we have $\lambda_{\wtilde{G}}(t,s)\geq k'$. Because $(L,R)$ has size $k'$, $(L,R)$ is a $t$-$s$ mincut and $t\in T^{*}$.

\item Consider property \ref{prop:Cover2}. By \Cref{lemma:Induction}, at the end of the algorithm, $\Gamma = T^{*}$. Property \ref{prop:Cover2} follows the update rule of $\Gamma$ (i.e. \Cref{line:UpdateGamma}). 

\item Property \ref{prop:Cover3} is because \Cref{thm:maxiso} returns a collection of disjoint cuts.
\end{itemize}
Lastly, the width of ${\cal K}$ is bounded by the number of inner loops, which is at most $O(k\log^{3}n)$.

Regarding the running time, the bottleneck is calling \Cref{thm:maxiso}. We invoke \Cref{thm:maxiso} $\OO(k\log^{3}n)$ times, each of which takes $m^{1+o(1)}$ time. Hence the total running time is $\OO(km^{1+o(1)})$.

\end{proof}

\section{Omitted Proof}
\label{sect:OmittedProof}
\begin{proof}[Proof of \Cref{ob:BetterWitness}]
Our strategy is to improve $(L,R)$ to $(L',R')$ with $(L'\setminus R')\cap T = (L\setminus R)\cap T$ with additional guarantees: (1) each connected component $C$ of $G[L'\setminus R']$ intersects $T$, and (2) $N_{G}(L'\setminus R') = L'\cap R'$. As we show next, this is indeed possible. Then for each $v\in L'\cap R'\cap T$, $v\in N_{G}(C)$ for some component $C$ of $G[L'\setminus R']$. Furthermore, because $C$ intersects $T$, we have $v\in N_{H}(C\cap T)$, which implies $v\in N_{H}((L'\setminus R')\cap T)$ and we are done. 

We can obtain $(L',R')$ with the two additional guarantees as follows. First, let $\Gamma$ denote the union of all components in $G[L\setminus R]$ that are disjoint from $T$. The cut $(L_{1},R_{1})$ with $L_{1} = L\setminus \Gamma$ and $R_{1} = R\cup \Gamma$ is obviously a witness satisfying the first extra property, and we have $(L_{1}\setminus R_{1})\cap T = (L\setminus R)\cap T$. Next, consider the cut $(L_{2},R_{2})$ with $L_{2} = (L_{1}\setminus R_{1})\cup N(L_{1}\setminus R_{1})$ and $R_{2} = V(G)\setminus (L_{1}\setminus R_{1})$. It is still a witness satisfying both the first and second extra properties, and trivially $(L_{2}\setminus R_{2})\cap T = (L_{1}\setminus R_{1})\cap T$. Finally, take $(L_{2},R_{2})$ as the desired $(L',R')$.
\end{proof}

\section{Connection to Tree Decomposition with Bounded Width}
\label{sec:connection}

Here, we argue that unbreakable decomposition generalizes tree decomposition with bounded width to general graphs. 

By definition, a graph $G$ has bounded treewidth if and only if there exists a tree decomposition of $G$ with bounded width, i.e., all bags have bounded size. So we cannot hope for this decomposition to exist in an arbitrary graphs. In contrast, for any $k\le q$, $(q,k)$-unbreakable decomposition exists for every graph \cite{cygan2020randomized}. 

We will show that, on a graph with bounded treewidth, every unbreakable decomposition is precisely a tree decomposition with bounded width. To prove this, we need the following well-known fact (see e.g.~\cite{reed1997tree}).
\begin{proposition}\label{fact:sep treewidth}Let $G$ be a graph with treewidth at most $k$. Let $X$ be any vertex set of size $s$. Then, there exists a vertex cut $(L,R)$ of size $k$ where $|L\cap X|,|R\cap X|>s/3$.
\end{proposition}

Now, we formally prove our claim.
\begin{corollary}
Let $G$ be a graph with treewidth at most $k$. Let $(T,\beta)$ be a $(q,k)$-unbreakable decomposition of $G$. Then, each bag of $T$ has size $|\beta(t)|\le3q$ for all tree node $t$. That is, the width of $T$ must be at most $3q$.
\end{corollary}

\begin{proof}
Suppose for contradiction that there is a tree node $t$ whose bag $\beta(t)$ contains more than $3q$ vertices. By \Cref{fact:sep treewidth}, there is a vertex cut $(L,R)$ of size $k$ where $|L\cap\beta(t)|,|R\cap\beta(t)|>q$. This contradicts that $\beta(t)$ is $(q,k)$-unbreakable. 
\end{proof}

\end{document}